\documentclass[12pt]{article}

\usepackage{a4wide}
\usepackage{amsmath,amsxtra}
\usepackage{amssymb}
\usepackage{txfonts}

\usepackage[all]{xy}

\newtheorem{theorem}{Theorem}[section]
\newtheorem{deff}[theorem]{Definition}

\newtheorem{proposition}[theorem]{Proposition}
\newtheorem{example}[theorem]{Example}
\newtheorem{lemma}[theorem]{Lemma}

\newtheorem{cor}[theorem]{Corollary}
\newtheorem{prop}[theorem]{Proposition}

\numberwithin{equation}{section}

\newcommand{\proof}{\vskip2mm \noindent {\bf Proof.}~}
\newcommand{\End}{\mathop{\mathrm{End}}}

\newcommand{\rt}{\rightarrow}

\newcommand{\mto}{\mapsto}
\newcommand{\na}{\nabla}

\newcommand{\Mor}{\mathrm{Mor}}
\newcommand{\Maps}{\mathrm{Map}}

\newcommand\beqa {\begin{eqnarray}}
\newcommand\eeqa {\end{eqnarray}}
\newcommand\bqa {\begin{eqnarray}}
\newcommand\eqa {\end{eqnarray}}
\newcommand{\beq}{\begin{eqnarray}}
\newcommand{\beqn}{\begin{eqnarray}\nonumber}
\newcommand{\eeq}{\end{eqnarray}}
\newcommand{\be}{\begin{array}}
\newcommand{\ee}{\end{array}}

\newcommand\bea {\begin{eqnarray}}
\newcommand\eea {\end{eqnarray}}

 \newcommand{\pr}{\partial}

 \newcommand{\cM}{{\cal M}}

 \newcommand{\cF}{{\cal F}}

\newcommand{\N}{{\mathbb N}}

\newcommand{\R}{{\mathbb R}}
\newcommand{\Z}{{\mathbb Z}}

 \newcommand{\g}{{\mathfrak g}}

\newcommand{\md}{\mathrm{d}}

\newcommand{\rank}{\mathop{\mathrm{rank}}}
 \def\S{{\Sigma}}
\def\2{{\textstyle\frac{1}{2}}}
\def\ba{\begin{eqnarray}}
\def\ea{\end{eqnarray}}

\def\CG{{\cal G}}
\def\CL{{\cal L}} \def\CX{{\cal X}} \def\CA{{\cal A}} 
  \def\rd{\mathrm{ d}}
   \def\CM{{\cal M}}

  \def\C{{\cal C}}
\def\CR{{\cal R}}

\def\CQ{{\cal Q}}

\def\Emd{\!\,{}^E\!\md}
\def\EL{\,\!{}^E\!L}

\def\Eg{\!\,{}^E\!\mathrm{g}}
\def\Vg{\!\,{}^V\!\mathrm{g}}
\def\Ena{{{}^E{}\nabla}}
\def\GE{\Gamma(E)}

\def\o{\omega}

\def\OE{\Omega_E (M)}

\makeatletter
\def\bard{\protect\@bard}
\def\@bard{%
\relax
\bgroup
\def\@tempa{\hbox{\raise.73\ht0
\hbox to0pt{\kern.4\wd0\vrule width.7\wd0
height.1pt depth.1pt\hss}\box0}}%
\mathchoice{\setbox0\hbox{$\displaystyle\mathrm{d}$}\@tempa}%
{\setbox0\hbox{$\textstyle \mathrm{d}$}\@tempa}%
{\setbox0\hbox{$\scriptstyle \mathrm{d}$}\@tempa}%
{\setbox0\hbox{$\scriptscriptstyle \mathrm{d}$}\@tempa}%
\egroup
}
\makeatother

\makeatletter
\def\barp{\protect\@barp}
\def\@barp{%
\relax
\bgroup
\def\@tempa{\hbox{\raise.73\ht0
\hbox to0pt{\kern.4\wd0\vrule width.7\wd0
height.1pt depth.1pt\hss}\box0}}%
\mathchoice{\setbox0\hbox{$\displaystyle\partial$}\@tempa}%
{\setbox0\hbox{$\textstyle \partial$}\@tempa}%
{\setbox0\hbox{$\scriptstyle \partial$}\@tempa}%
{\setbox0\hbox{$\scriptscriptstyle \partial$}\@tempa}%
\egroup
}
\makeatother

\makeatletter
\def\barb{\protect\@barb}
\def\@barb{%
\relax
\bgroup
\def\@tempa{\hbox{\raise.73\ht0
\hbox to0pt{\kern-.1\wd0\vrule width.7\wd0
height.1pt depth.0pt\hss}\box0}}%
\mathchoice{\setbox0\hbox{$\displaystyle\mathrm{b}$}\@tempa}%
{\setbox0\hbox{$\textstyle \mathrm{b}$}\@tempa}%
{\setbox0\hbox{$\scriptstyle \mathrm{b}$}\@tempa}%
{\setbox0\hbox{$\scriptscriptstyle \mathrm{b}$}\@tempa}%
\egroup
}
\makeatother

   \def\C{\mathbb C}
   \def\a{\alpha}
   \def\b{\beta}
    
    \def\t{\theta}
   \def\de{\delta}

   \def\e{\epsilon}
   \def\i{\imath}

\def\Gr{\mathrm{G}} 

\def\G1{\Gr_1}

\def\Ad{\mathrm{Ad}}

\begin{document}

\begin{center}
{ \bf \Large  Generalizing Geometry---Algebroids and Sigma Models}
\end{center}


\begin{center}\sl\large
Alexei Kotov$^1$ and Thomas Strobl$^2$
\end{center}

\begin{center}{\small\sl
$^1$  Universite du Luxembourg,
Laboratoire de  Mathematiques, \\ Campus Limpertsberg, 162A,
Avenue de la Faiencerie, L-1511 Luxembourg\\
e-mail address: {\tt alexei.kotov{}@{}uni{}.{}lu}\\ \vspace{0.5em}
$^2$ Univ.~de Lyon, Univ.~Lyon1, CNRS UMR 5208, Institut Camille Jordan, \\43 Boulevard du 11 Novembre 1918, F-69622 Villeurbanne Cedex, France\\
email address: {\tt strobl{} {}@{} {}math.univ-lyon1.fr}
}
\end{center}



\begin{abstract} In this contribution we review some of the interplay between sigma models in theoretical physics and novel geometrical structures such as Lie (n-)algebroids. The first part of the article contains the mathematical background, the definition of various algebroids as well as of Dirac structures, a joint generalization of Poisson, presymplectic, but also complex structures. Proofs are given in detail. The second part deals with sigma models. Topological ones, in particular the AKSZ and the Dirac sigma models, as generalizations of the Poisson sigma models to higher dimensions and to Dirac structures, respectively, but also physical ones, that reduce to standard Yang Mills theories for the ``flat'' choice of a Lie algebra: Lie algebroid Yang Mills theories and possible action functionals for nonabelian gerbes and general higher gauge theories.
Characteristic classes associated to Dirac structures and to higher principal bundles are also mentioned.
\end{abstract}

\setcounter{tocdepth}{1}
\tableofcontents

\smallskip

\section{Introduction}

In this article we want to summarize some ideas in the overlap
of differential geometry and mathematical physics. In particular we
focus on the interplay of so-called sigma models with geometrical
structures being related to algebroids in one way or another.
Several traditional geometrical notions received various kinds of
generalizations in recent years. Some of them give rise to novel sigma
models, while also some sigma models bring attention to possible new
geometries.

Sigma models are action functionals (variational problems) where
the underlying space of fields (maps) has a target space equipped
with some geometry; in the most standard case one regards maps $\CX$ from
one given Riemannian manifold $(\Sigma,h)$ to another one $(M,g)$ and
considers the functional
\beq S[\CX] = \frac{1}{2} \int_\S || \rd \CX ||^2 \, ,
\label{standardsigma} \eeq
where $||\a||^2 = (\CX^* g)(\a \! \stackrel{\wedge}{,} \! * \a)$ for any $\a
\in \Omega^p(\Sigma, \CX^* TM)$ with $*$ denoting the $h$-induced
Hodge duality operation on $\S$,\footnote{In local coordinates
$\sigma^\mu$ on $\Sigma$ and $x^i$ on $M$ this expression reads more
explicitly as $\a_{i\mu_1 \ldots \mu_p} \a^{i\mu_1 \ldots \mu_p} \rd
{\mbox{vol}}_\Sigma$ where indices $\mu_i$ and $i$ are raised and
lowered by means of the metric $h$ and (the pullback by $\CX$ of) $g$,
respectively, and $\rd {\mbox{vol}}_\Sigma= \sqrt{\det(h)}
\rd^d \sigma$, $d$ denoting the dimension of $\S$. In this article we
use the Einstein sum convention, i.e.~a sum over repeated indices is
always understood.} the critical points of which are precisely the
maps $\CX$ which are harmonic. For the special case that $M$ is just
$\R^n$ equipped with the standard flat metric this functional reduces
to
\beq S[\phi^i] =  \frac{1}{2} \int_\S \rd \phi^i \wedge *   \rd \phi^i
\, , \label{eq:freescalar}
\eeq
the action of $n$ ``free scalar fields'' on $\S$ (here $\phi^i$
denotes the function on $\S$ obtained by restricting $\CX$ to the i-th
coordinate in $\R^n$---it is understood that the index $i$ on the
l.h.s.~is an ``abstract index'', i.e.~$S$ depends on all the scalar
fields, $\phi^1$ to $\phi^n$).

Another example of a sigma model is the Poisson sigma model
\cite{PSM,Ikeda}, where the source manifold $\S$ is necessarily
two-dimensional and the target manifold $M$ carries a Poisson
structure instead of a Riemannian one. In fact, one considers a
functional on the space of vector bundle morphisms from $T\S$ to
$T^*M$ in this case. It is ``topological'', which we want to interpret
as saying that the space of classical solutions (stationary points of
the functional) modulo gauge transformations (invariances or
symmetries of the functional) is finite dimensional---besides the fact
that it does not depend on geometrical structures of the source like a
Riemannian metric $h$ (which in this case is even absent in the
definition of the functional).  The tangent bundle of any manifold as
well as the cotangent bundle of a Poisson manifold give rise to what
are called Lie algebroids (whose definition is properly recalled in
the body of the paper below---cf.~in particular Examples \ref{ex:TM}
and \ref{Poisson1} below) and the above functional stationarizes
precisely on the morphisms of these Lie
algebroid structures \cite{BKS}.

But this also works the other way around: given a functional of such a
form defined by an a priori arbitrary bivector field $\Pi$ (and, in
some particular extension of the functional, also a closed 3-form $H$)
the respective functional becomes topological, \emph{iff} \cite{PSM0,Ctirad}
$\Pi$ defines a (for non-zero $H$ twisted) Poisson structure. In fact,
twisted Poisson structures (c.f.~also \cite{SeveraWeinstein}), i.e.~bivectors
$\Pi$ together with closed 3-forms $H$ defined over a manifold $M$
satisfying
\beq  [\Pi ,\Pi]=(\Pi^\natural)^{\otimes 3}H \label{HPoisson} \eeq
where $[ ... ]$ denotes the Schouten-Nijenhuis bracket of multi-vector
fields and $\Pi^\natural \colon T^*M\to TM$ is the natural operator induced by
the bivector $\Pi$ as follows: $\Pi^\natural (\a)=\Pi (\a, \cdot)$,
 were even \emph{found}
first in such a manner \cite{Ctirad,Park}. This is typical for the
interplay of geometrical notions and sigma models: the former are
needed to define the latter, but sigma models sometimes also give
indications about (focus on) particularly interesting geometrical
notions. In this example, in the space of bivectors on a manifold the
ones which are Poisson are singled out by the sigma model, or in the
space of pairs $(\Pi,H)$ those satisfying eq.~(\ref{HPoisson}) (which
can be seen to define a particular Dirac structure
\cite{SeveraWeinstein}---we introduce to Dirac structures in the main text in
detail).  Another example for such an interplay are
supersymmetric sigma models and bihermitian geometry
\cite{GatesHullRocek}. The latter geometry received renewed and
revived interest recently by its elegant reformulation in terms of
so-called generalized complex structures \cite{HitchinGCS}.

What is, on the other hand, meant more specifically by generalizing
traditional geometrical notions? In fact, also the generalized notions
can usually be expressed in terms of ordinary differential geometrical
ones, in which case it  just boils down to a different way of thinking
about them. In any such case one has usually some particular kind of a
so-called \emph{algebroid} in the game.

There are several kinds of an algebroid considered in the
literature.   All of them have the following data in common, which we
thus want to use as a definition of the general term:
\begin{deff} \label{algebroid}
We call an {\rm{algebroid}} $E=(E,\rho,[ \cdot , \cdot ])$, a vector bundle
$E \to M$ together with a homomorphism of vector bundles $\rho \colon E \to TM$, called
the {\rm{anchor}} of $E$, and a product or bracket on the sections of
$E$ satisfying the Leibniz rule ($\psi,\psi' \in \Gamma(E)$, $f \in C^\infty(M)$):
\beq
[\psi, f \psi'] = f [\psi,  \psi'] + \left(\rho(\psi) f \right) \,
\psi' \,.
\label{Leibniz}
\eeq
\end{deff}
Note that the map on sections induced by $\rho$ is denoted by the same
letter conventionally.  Depending on further conditions placed on the
bracket (like its symmetry properties or its Jacobiator) and further structures defined
on $E$, one has different kinds of algebroids: Lie algebroids, Courant
algebroids, strongly homotopy Lie (or $L_\infty$-) algebroids, and so
on.
Lie algebroids are obtained e.g.~by requiring in addition that the bracket is
a Lie bracket, i.e.~antisymmetric and with vanishing Jacobiator. The
definition of a Courant algebroid requires a fiber metric on $E$
controlling the symmetric part of the bracket as well as its Jacobiator.
We will come back to all these various kinds of algebroids in more
detail in the text below.

The philosophy now is that we can do differential geometry by
replacing $TM$ or, more generally, also the tensor bundle $\tau^p_q(M)
= TM^{\otimes p} \otimes T^*M^{\otimes q}$ by $E$ and $E^{\otimes p}
\otimes E^{*\otimes q}$, respectively. 
Let us call a section $t$ of $\oplus_{p,q}
\left(E^{\otimes p} \otimes E^{*\otimes q}\right)$ an $E$-tensor
field. The Leibniz property permits us to define a (Lie)
``derivative'' of $t$ along any section $\psi$ of $E$: Indeed, set
$\EL_\psi(f) :=
\rho(\psi) f$,  $\EL_\psi (\psi') := [\psi,\psi']$, and extend this to
powers of $E$ by the Leibniz property w.r.t.~tensor multiplication
and to $E^*$ by means of compatibility with contraction: $\EL_\psi
\left(\langle \psi' , \omega \rangle\right) = \langle \EL_\psi(\psi') ,
\omega \rangle + \langle \psi' , \EL_\psi( \omega) \rangle$, defining
the Lie derivative for any $\omega \in \Gamma(E^*)$.  This
implies that given an algebroid $E$ (not necessarily Lie), i.e.~data
that at least include those of the definition above, we can define a
``Lie derivative'' of any $E$-tensor along sections of $E$. By
construction, this generalizes the notion of an ordinary Lie
derivative of ordinary tensor fields: indeed, the usual formulas are
reproduced in the case where the algebroid is chosen to be a so-called
standard Lie algebroid, i.e.~where $E$ is the tangent bundle of the
base manifold $M$, the anchor $\rho$ is the identity map and the
bracket is the standard Lie bracket of vector fields. Such type of
geometrical notions have properties in common with their prototypes in
traditional geometry (like, in this case e.g., by construction, the
Leibniz property of the generalized Lie derivatives), but they in
general also have pronounced differences, possibly depending on the
type of algebroid, i.e.~on the additional structures. In the above
example one can ask e.g.~if the commutator of such Lie derivatives is
the Lie derivative of the bracket of the underlying sections. In
general this will not always be the case. However, for an important
subclass, containing Lie and Courant algebroids, it will (cf.~Lemma
\ref{lemma:rep} below).

Another important example of such generalized structures is
``generalized geometry'' or \emph{generalized complex geometry} in the
sense of Hitchin (cf.~also the contribution of N.~Hitchin to this
volume). It is a particular case of the above viewpoint where $E$ is
taken to be $TM \oplus T^*M$ equipped with projection to the first
factor as anchor as well as the so-called Courant or Dorfman bracket
(in fact, slightly more general and conceptually preferable, one takes
$E$ to be what is called an exact Courant algebroid---cf.~definition
\ref{def:exact} below). Now, by definition a generalized complex
structure is what usually would be a complex structure, just replacing
the standard Lie algebroid $TM$ by an exact Courant algebroid. In
particular, it is an endomorphism of $E \cong TM \oplus T^*M$ squaring
to minus one and satisfying an integrability condition using the
bracket on $E$ (cf., e.g., Prop.~\ref{prop:Nij} below). It is not
difficult to see that this notion generalizes simultaneously ordinary
complex structures as well as symplectic ones. In fact, the situation
is closely related to (real) Dirac structures (particular Lie subalgebroids of
exact Courant algebroids, cf.~definition \ref{def:Dirac} below),
mentioned already previously above: these generalize simultaneously
Poisson and presymplectic structures on manifolds. In fact,
generalized complex structures did not only find at least part of
their inspiration from real Dirac structures, but they can be even
defined equivalently as imaginary Dirac structures---which is the
perspective we want to emphasize in the present note.\footnote{This
point of view was maybe less known at the time when we this note was
started, while in the mean time it has received some attention also
elsewhere.}

An elegant and extremely useful viewpoint on some algebroids arises within
the language of differential graded manifolds, sometimes also called
Q-manifolds ($Q$ denoting a homological degree one vector field on the
graded manifold, i.e.~its differential). We devote a section to
explaining this relation, after having introduced the reader to the
above mentioned notions of algebroids, Dirac structures, and
generalized complex structures, in the three sections to
follow. Together these four sections provide our exposition on
algebroids and this kind of generalized geometry.

Some of the following sections then deal with the respective sigma
models: Given a Q-manifold with a compatible graded symplectic
structure on it, one can always associate a topological sigma model to
it \cite{AKSZ}. We review this construction in some detail and specialize it
to lowest dimensions (of the source manifold $\Sigma$), reproducing
topological models corresponding to Poisson manifolds ($\dim \S =2$,
this is the above mentioned Poisson sigma model), to Courant
algebroids ($\dim \S =3$, such models were considered in \cite{RoytenbergAKSZ}).
Some space is devoted to describing these models, somewhat complementary to
what is found in the literature, since
they can be used to introduce part of the formalism that is needed for the
last section to this contribution.

There are also topological models that, at least up to now,
have not yet been related to the AKSZ formalism and corresponding
e.g.~to Dirac structures. We recall these models, called Dirac sigma
models \cite{KSS} and generalizing the Poisson sigma models
essentially such as (real) Dirac structures generalize Poisson
manifolds,  in a separate section. These  as well as the AKSZ models share
the property that they are topological and that the solutions to their field
equations generalize (only) \emph{flat} connections to the algebroid setting.

The final section is devoted to sigma models in
arbitrary spacetime dimension (dimension of $\S$ which then is taken
to be pseudo-Riemannian) with a relation to algebroids but which are
nontopological and which generalize connections that are not necessarily flat but instead satisfy the Yang-Mills field equations.\footnote{Being nontopological is required for most
physical applications so as to host the degrees of freedom necessary
to describe realistic interactions. --- There are also another type of (nontopological) sigma models than those explained in the last section that are related to algebroids and in particular generalized and bihermitian geometry. These are  supersymmetric two-dimensional sigma
models, i.e.~\emph{string theories}. Although also highly interesting, we will not touch this issue here but refer e.g.~to the review article \cite{Zabzine-review} and references therein.} This deserves some further
motivating explanation already in the introduction:

Consider replacing in eq.~(\ref{eq:freescalar}) the functions on $\S$
by 1-forms $A^a$, $a=1,\ldots,r$, yielding
\beq
S[A^a] =  \frac{1}{2} \int_\S \rd A^a \wedge *   \rd A^a
\, . \label{eq:free1forms}
\eeq
For $r=1$, i.e.~there is just one 1-form field $A$, this is a famous
action functional, describing the electromagnetic interactions (the
electric and the magnetic fields can be identified with the components
of the ``field strength'' 2-form $\rd A$). Having several such 1-forms
in the game, $r>1$, one obtains the functional describing $r$ free
(i.e.~mutually independent) 1-form fields\footnote{In the physics
language they are often called ``vector fields'' (as opposed to
``scalar fields'' used in (\ref{eq:freescalar})). We avoid this
somewhat misleading/ambiguous nomenclature, but we will, however, from
time to time refer to them as ``gauge fields'', despite the fact that
in the mathematical setting they correspond to connections in a principal
$G$-bundle (here trivialized with $G=U(1)^r$).}. The most standard way of
making scalar fields interact is to add some ``potential term'' to the
functional (\ref{eq:freescalar}), i.e.~to add the integral over
$V(\phi^i)$ multiplied by the $h$-induced volume form on $\S$ (where
$V$ is some appropriately smooth function on $\R^n$, mostly even only
a low degree polynomial so as to not spoil
``renormalizability'').\footnote{An alternative way of having scalar
fields interact is coupling them to 1-form gauge fields so that they
start being correlated (i.e.~interacting) via these 1-form fields. In
fact, both ways of interactions are realized in the standard model of
elementary particle physics, where the gauge fields describe
interaction particles like the photon and the scalar fields describe
``matter'', essentially like electrons (or the--not yet discovered--Higgs particles).}
Turning (\ref{eq:free1forms}) into an interacting
theory (without introducing further fields and not spoiling its gauge
invariance, at most ``deforming'' the latter one appropriately) is not
so simple. In fact, the result is rather restricted (cf., e.g.,
uniqueness theorems in the context of the deformation theory of gauge
theories \cite{BarnichBrandtHenneaux-Review}) and one is lead to only
replace $\rd A^a$ by the expression
\beq F=\rd A + \frac{1}{2}[ A \! \stackrel{\wedge}{,} \! A] \label{F} \eeq
for the curvature of the Lie algebra valued
connection 1-form $A$ of a (trivialized) principal bundle for some
$r$-dimensional, quadratic Lie group $G$. ``Quadratic'' means that its
Lie algebra $\g$ admits an ad-invariant inner product $\kappa$ which,
when starting by deforming (\ref{eq:free1forms}), needs to have
definite signature such that a sum over the index $a$ results from a term of the
form
\beq S_{YM}[A]=\int_\S \kappa(F\! \stackrel{\wedge}{,} \! *F) \label{eq:YM} \eeq
after choosing a $\kappa$-orthonormal basis in $\g$. The resulting
theory 
is called a Yang-Mills (gauge) theory and was found to govern all the strong,
the weak, and the electromagnatic interactions.

If one considers an algebroid $E$ as defined above over a zero
dimensional base manifold $M$, $M$ degenerating to a point, one is
left only with a vector space. For $E$ being a Lie algebroid this
vector space becomes a Lie algebra, for $E$ being a Courant algebroid
it becomes a quadratic Lie algebra. In fact, this is a second,
algebraic part that is incorporated in algebroids: general Lie
algebroids can be thought of as a common generalization of the
important notions of, on the one hand, a Lie algebra $\g$ and, on the
other hand, standard geometry (i.e.~geometry defined for $TM$ or
$\tau_q^p(M)$), as we partially explained already above. (A similar
statement is true for general Courant algebroids, $TM$ then being
replaced by the ``standard Courant algebroid'' $TM\oplus T^*M$ and
$\g$ by a quadratic Lie algebra together with its invariant scalar
product). From this perspective it is thus tempting to consider the
question if one can define e.g.~a theory of principal bundles with
connections in such a way that the structural Lie algebra is replaced
by (or better generalized to) appropriately specified structural
algebroids. Likewise, from the more physical side, can one
generalize a functional such as the Yang-Mills functional
(\ref{eq:YM}) to a kind of sigma model, replacing the in some sense
flat Lie algebra $\g$ by nontrivial geometry described via some
appropriate Lie or Courant algebroid?  These questions will be
addressed and, at least part of them, answered to the positive in the
final section to this article.

Between the part on sigma models and the one on algebroids it would have been nice to also include a section on current algebras (cf.~\cite{AS} as a first step), as another link between the two; lack of spacetime, however, made us decide to drop this part in the present contribution.

\section{Lie and Courant algebroids}
\label{sec:GG}

In the present section, we recall the notions of Lie and Courant
algebroids in a rigorous manner and study some of their properties. We also briefly introduce some of
their higher analogues, like Lie 2-algebroids and vector bundle
twisted Courant algebroids. In a later section we will provide
another, alternative viewpoint of all these objects by means of graded manifolds,
which permits an elegant and concise reformulation.


\begin{deff} \label{Lodayalgebroid}
A \emph{Loday algebroid} is an algebroid $(E,\rho, [
\cdot , \cdot])$ as defined by Definition \ref{algebroid}
where the bracket defines a Loday algebra on $\Gamma(E)$, i.e.~it
satisfies the Loday (or left-Leibniz) property,
\beq [\psi_1 , [\psi_2,\psi_3]]=[ [\psi_1 , \psi_2],\psi_3] +
[\psi_2 , [\psi_1,\psi_3]] \, . \label{Loday}
\eeq
An \emph{almost Lie algebroid} is an algebroid $E$ where the bracket
is antisymmetric. $E$ becomes a \emph{Lie algebroid}, if $[ \cdot ,
\cdot]$ defines a Lie algebra structure on $\Gamma(E)$,
 i.e.~if $E$ is simultaneously Loday and almost Lie.
\end{deff}
Here we adapted to the nomenclature of Kosmann-Schwarzbach, who
prefers to use the name of Loday in the context of (\ref{Loday}) so as
to reserve the terminus Leibniz for compatibility
w.r.t.~multiplication of sections by functions, (\ref{Leibniz}), for
which ``Leibniz rule'' has become standard.

\begin{example} \label{TM} Obviously, $(TM,\rho=Id)$ together with the Lie bracket of
vector fields is a Lie algebroid; it is called the \emph{standard} Lie algebroid.

If $M$ is a point, on the other
hand, a Lie algebroid reduces to a Lie algebra. More generally, if
the anchor of a Lie algebroid map vanishes (i.e.~maps $E$ to
the image of the zero section in $TM$), $E$ is a bundle of Lie
algebras; in general not a Lie algebra bundle, since Lie algebras of
different fibers need not be isomorphic.\end{example}

\begin{example} \label{Poisson1} A less trivial Lie algebroid
is the cotangent bundle $T^*M$ of a Poisson manifold.\footnote{By
definition, one obtains a Poisson structure on a smooth manifold $M$,
if the space of functions is equipped with a
Lie bracket $\{ \cdot ,\cdot\}$ satisfying $\{f,g h\} = g \{ f, h\} +
h \{ f, g\}$ for all $f,g,h \in C^\infty(M)$. The latter condition, together with the antisymmetry of the bracket, is equivalent to the existence of a
bivector field $\Pi \in \Gamma(\Lambda^2 TM)$ such that
$\{f,g\}=\langle \Pi, \md f \otimes \md g
\rangle$.} The anchor is provided by contraction with the Poisson bivector field
$\Pi$ and the bracket of exact 1-forms $[\md f, \md g] := \md
\{f,g\}$ is extended to all 1-forms by means of the Leibniz rule
(\ref{Leibniz}).
\end{example}

\begin{example} An example of a Loday algebroid with a non-antisymmetric
bracket is the following one: $E=TM \oplus T^*M$ with $\rho$ being
projection to the first factor and the bracket being given by
\beq  [\xi +\a ,\xi' +\a']=[\xi ,\xi']+L_\xi\a' -\i_{\xi'}\md\a \, ;
\label{Courant-bracket}
\eeq
here $\xi,\xi'$ and $\a,\a'$ denote vector fields and 1-forms on $M$,
respectively, and the bracket on the r.h.s.~is the usual Lie bracket of
vector fields. This is the socalled Dorfmann bracket. Note
that if one takes the antisymmetrization of this bracket (the original
bracket Courant has introduced \cite{Courant}), $E$ does \emph{not}
become an algebroid in our sense, since eq.~(\ref{Leibniz}) will not
be valid any more (as a consequence of the non-standard behavior of the
above bracket under multiplication by a function of the section in the
first entry).
\end{example}

\begin{lemma} \label{morphism}
The anchor map of a Loday algebroid is a morphism of brackets.
\label{lem:morphism}
\end{lemma}

\begin{proof} Obvioulsy, from (\ref{Loday}) we find
$[ [\psi_1 , \psi_2],f\psi_3] =  [\psi_1 , [\psi_2,f\psi_3]]  -
\left(\psi_1 \leftrightarrow \psi_2\right)$.
Using the Leibniz rule for the l.h.s., we obtain
$\rho([\psi_1,\psi_2])f \, \psi_3 + f [ [\psi_1 ,
\psi_2],\psi_3]$. Applying it twice to the first term on the r.h.s., we get
$\rho(\psi_1)\rho(\psi_2) f \, \psi_3  + f  [\psi_1 , [\psi_2,\psi_3]]
+ \rho(\psi_1)f \, [\psi_2,\psi_3] +  \rho(\psi_2)f \,
[\psi_1,\psi_3]$. The last two terms drop out upon antisymmetrizing in
$\psi_1$ and $\psi_2$. In the remaining equation the terms proportional
to $f$ cancel by means of (\ref{Loday}), and one is left with
\beq \rho ([\psi_1,\psi_2])f \, \psi_3 =[\rho(\psi_1),\rho(\psi_2)] f \,
\psi_3 \, , \eeq
valid for all sections $\psi_i$, $i=1,2,3$, and functions $f$. This
completes the proof. $\square$
\end{proof}

\begin{lemma} \label{lemma:rep} The $E$-Lie derivative provides a representation of
the bracket of a
Loday algebroid on $E$-tensors,
$[\EL_{\psi_1},\EL_{\psi_2}]=  \EL_{[\psi_1,\psi_2]}$                .
\end{lemma}
\begin{proof}
This follows from (\ref{Loday}), the previous Lemma, and the extension
of the E-Lie derivative to tensor powers and the dual by means of a
Leibniz rule (using that commutators of operators satisfying a Leibniz
rule for some algebra---here that of the tensor product as well as
that of contractions---are of Leibniz type again).  $\square$
\end{proof}

A Lie algebroid permits to go further in extrapolating usual
geometry on manifolds to the setting of more general vector
bundles (algebroids). In particular, there is a straightforward
generalization of the de Rham differential in precisely this
case:\footnote{We remark in parenthesis that there is also an
option to generalize the de Rham differential different from formula
(\ref{deRham}) below, using the language of graded
manifolds; \emph{that}
generalization can be used also for Courant algebroids
(cf.~e.g.~\cite{Roytenberg}) or even arbitrary $L_\infty$
algebroids, cf.~section \ref{Qman} below.}
In any almost Lie algebroid we may define the following degree one map
$\Emd$ on $E$-differential forms $\OE \equiv \Gamma(\Lambda^\cdot
E^*)$. For any function $f$ and $E$-1-form $\omega$ we set
\beq\label{diff-can} \langle \Emd f , \psi \rangle := \rho(\psi) f
\; , \, \langle \Emd \omega , \psi \otimes \psi' \rangle :=\rho
(\psi)\langle \omega , \psi'\rangle -\rho (\psi')\langle \omega ,
\psi\rangle -\langle\omega , [\psi ,\psi']\rangle\label{deRham}
\eeq and extend this by means of a graded Leibniz rule to all of
$\Omega_E (M) $. Clearly, for the standard Lie algebroid $(TM, Id)$ this
reduces to the ordinary de Rham differential. As one proves by
induction, with this one finds in generalization of
the Cartan-Koszul formula:
\beq\label{Cartan-Koszul} \Emd \omega
(\psi_1,...,\psi_{p+1})&:=& \sum\limits_{i=1}^{p+1}(-1)^{i+1} \rho
(\psi_i)\omega (...,\hat{\psi_i},...)+
\\ \nonumber &&
+\sum\limits_{i<j}(-1)^{i+j}\omega ([\psi_i,\psi_j]...,
\hat{\psi_i},...,\hat{\psi_j},...)\;,
\eeq
valid for any $\omega\in\Omega_E^p (M)$ and $\psi_i\in\Gamma (M, E)$.
This text is in part supposed to contain explicit proofs:

\begin{proof}(of eq.~(\ref{Cartan-Koszul}))
The property holds for $p=1$ by definition. Suppose
we proved (\ref{Cartan-Koszul}) for all forms of order at most
$p-1$. It suffices to see that (\ref{Cartan-Koszul}) holds for each
$\omega =\alpha\wedge \tau$, where $\alpha\in\Omega_E^1 (M)$,
$\tau\in \Omega_E^{p-1} (M)$. $\Emd (\omega)=\Emd
(\alpha)\wedge\tau -\alpha\wedge\Emd (\tau)$, therefore
\beqn  &&\Emd (\omega)(\psi_1,...,\psi_{p+1})=\sum\limits_{i<j}
(-1)^{i+j+1} \;\Emd \a (\psi_i ,\psi_j)\tau
(...,\hat{\psi_i},...,\hat{\psi_j},...)+
\\&&\nonumber \sum\limits_{i} (-1)^i \a (\psi_i)\;\Emd\tau (...,
\hat{\psi_i},...) = \sum\limits_{i<j} (-1)^{i+j+1} \;
\left(\rho(\psi_i)\a (\psi_j)- \rho (\psi_j)\a (\psi_i) -\right.
\\\nonumber &&\left.\a
([\psi_i,\psi_j])\right)\! \tau
(...,\hat{\psi_i},...,\hat{\psi_j},...)\! +\!\!\sum\limits_{i}
(-1)^i \a (\psi_i)\left(\sum\limits_{j<i}(\!-\!1)^{j+1}\rho
(\psi_j)\tau (..., \hat{\psi_j},...,\hat{\psi_i},... )\right.
\\\nonumber &&+\left. \sum\limits_{i<j}(-1)^{j}\rho (\psi_j)\tau
(..., \hat{\psi_i},...,\hat{\psi_j},...
)+\sum\limits_{{k<l<i}\atop{i<k<l}} (-1)^{k+l}\tau
([\psi_k,\psi_l],...,\widehat{\psi_{i,k,l}},...)
+\right.\\\nonumber &&\left. + \sum\limits_{k<i<l}
(-1)^{k+l+1}\tau
([\psi_k,\psi_l],...,\widehat{\psi_{i,k,l}},...)\right)\;,\eeq
Collecting all terms proportional to $\rho$, we find
\beqn  && \sum\limits_{i<j} (-1)^{i+j+1} \;\left(\rho
(\psi_i)\a (\psi_j)- \rho (\psi_j)\a (\psi_i) \right) \tau
(...,\hat{\psi_i},...,\hat{\psi_j},...)+ \\ \nonumber &&
\sum\limits_{i} (-1)^i
\a (\psi_i) \left(
\sum\limits_{j<i}(-1)^{j+1}\rho (\psi_j)\tau (...,
\hat{\psi_j},...,\hat{\psi_i},... )+\sum\limits_{i<j}(-1)^{j}\rho
(\psi_j)\tau (..., \hat{\psi_i},...,\hat{\psi_j},... )
\right)=\\\nonumber & = & \sum\limits_{i=1}^{p+1}(-1)^{i+1} \rho
(\psi_i)\omega (...,\hat{\psi_i},...)\; , \eeq
while the remaining terms yield
\beqn && \sum\limits_{i<j} (-1)^{i+j} \a ([\psi_i,\psi_j])
\tau (...,\hat{\psi_i},...,\hat{\psi_j},...)+\\\nonumber &&
\sum\limits_{i}
(-1)^i \a (\psi_i)
\left(
\sum\limits_{{k<l<i}\atop{i<k<l}} (-1)^{k+l}\tau
([\psi_k,\psi_l],...,\widehat{\psi_{i,k,l}},...) +
\sum\limits_{k<i<l} (-1)^{k+l+1}\tau
([\psi_k,\psi_l],...,\widehat{\psi_{i,k,l}},...)\right)=\\\nonumber
&=& \sum\limits_{i<j}(-1)^{i+j}\omega ([\psi_i,\psi_j]...,
\hat{\psi_i},...,\hat{\psi_j},...)\;\eeq
which completes the proof.
$\square$ \end{proof}

One may  now verify that the above operator squares to
zero, $\Emd ^2=0$, \emph{iff} (\ref{Loday}) is satisfied
(turning $E$ into a Lie algebroid).
In fact, \emph{all} the information of a Lie
algebroid is captured by $\Emd$ as seen from the following

\begin{prop} The structure of a Lie algebroid on a vector bundle $E
\to M$ is in one-to-one correspondence with the structure of a
differential complex on $(\OE,\wedge)$. \label{prop:Lie}
\end{prop}
\begin{proof}
As obvious from (\ref{diff-can}), the skew-symmetric bracket
on sections of $E$
and the anchor map are uniquely determined by the canonical
differential $\Emd$. In particular, (\ref{diff-can}) can be rewritten
as
 \beq \label{rho-derived}\rho
(\psi )f&=&[\imath_\psi ,\Emd]f\;,\\
\label{bracket-derived}\i_{[\psi,\psi']}\omega&=&
\left[\i_{\psi},[\i_{\psi'},\Emd]\right] \omega\;,
\eeq
 where $\imath_\psi$ denotes the contraction
with $\psi\in\Gamma (M,E)$ and $[,]$
the super bracket of super derivations of
the graded commutative algebra $\OE$. Note that (\ref{bracket-derived})
has been verified for $E$-1-forms $\omega$ (and it is trivially
satisfied if $\omega$ is a function); due to the required
Leibniz property of $\Emd$, however,
it then necessarily holds for arbitrary $\omega \in \OE$ (which is
generated by  $E$-1-forms and  functions). With these formulas
we now find
\beqn &&\left(\rho ([\psi_1
,\psi_2])-[\rho (\psi_1),\rho (\psi_2)]\right)f
=\i_{\psi_1}\i_{\psi_2}\Emd^2 f \;, \\\nonumber
 &&\left(\i_{\left[\psi_1,[\psi_2,\psi_3]\right]}+
\;c.\;p.\right)\omega
 = \i_{\psi_1}\i_{\psi_2}\i_{\psi_3}\Emd^2 \omega
-\left(\i_{\psi_1}\i_{\psi_2}\Emd^2 \left(\i_{\psi_3}\omega\right) +
\;c.\;p., \right)\eeq
valid for any  $f\in C^\infty
(M) $, $\omega\in\Omega^1_E (M) $, and $\psi_i\in\Gamma (M,E)$.
Thus we see that the anchor map is a morphism of
brackets if and only if the canonical differential squares to zero
on functions, and the Jacobi condition for the almost Lie bracket
holds if and only if $\Emd^2=0$ on functions and $E$-1-forms
simultaneously. This extends to all of $\OE$ since $\Emd^2
=\frac{1}{2}[\Emd, \Emd]$ is a super derivative of the algebra
$\OE$. Noting that the Leibniz rule follows from the derivative property of $\Emd$
as well and that any Lie algebroid gives rise to such a differential finalizes
the proof. $\square$
\end{proof}

Eq.~(\ref{bracket-derived}) shows that the bracket of a Lie algebroid
is a socalled \emph{derived bracket} \cite{Kosmann}, generalizing
the wellknown formulas of Cartan: Eq.~(\ref{rho-derived}) generalizes
the fact that on differential forms one has $L_{\xi} = [\i_\xi , \md]$
(while on a function the Lie derivative reduces to application of the
vector field). Eq.~(\ref{bracket-derived}) reduces to the evident
identity $ -\i_{L_{\xi'}(\xi)}=[\i_\xi , L_{\xi'}]$ for the standard
Lie algebroid (using $L_{\xi'}(\xi)=-[\xi,\xi']$). We remark in
parenthesis that in the case of $E$ being a Lie algebroid the $E$-Lie
derivative, defined for a general Loday algebroid, satisfies also the
usual formula
\beq  \label{Lie} \EL_\psi = \i_\psi  \Emd  + \Emd \, \i_\psi \eeq
on $E$-differential forms. Later in this
article we also show how the Courant (or Dorfmann) bracket can be put into the form
of a derived bracket---for the case of eq.~(\ref{Courant-bracket})
already at the end of this section, while for the more general setting, which we
are going to define in what follows, in section \ref{Qman} below.


\begin{deff} \label{Cour}
A \emph{Courant algebroid} is a Loday algebroid
$(E,\rho,[ \cdot , \cdot])$ together with an invariant
$E$-metric $\Eg$ such that
$\Eg([\psi,\psi],\psi') =  \frac{1}{2}\rho(\psi')\Eg(\psi,\psi)$.
\end{deff}
We will mostly denote
$\Eg(\psi,\psi')$ simply as $(\psi,\psi')$, except if the appearance
of the metric shall be stressed. The $E$-metric is used to control the
symmetric part of the bracket. It permits to deduce the behavior
of the bracket under multiplication of the first section w.r.t.~a
function for example:
\beq  [f\psi_1,\psi_2] = f [\psi_1, \psi_2] -
\left(\rho(\psi_2)f\right) \psi_1 + (\psi_1,\psi_2)\,\rho^*(\md f)
\, . \eeq
In a general Loday algebroid there is
\emph{no} restriction to the bracket like this at all.

Invariance of $\Eg$ means that for any section
$\psi \in \Gamma(M,E)$ one has
$\EL_\psi \Eg =0$ ; this is the same as requiring
\beq ([\psi_1,\psi_2],\psi_3) + (\psi_2 , [\psi_1,\psi_3]) =
\rho(\psi_1)  (\psi_2,\psi_3) \label{ad}
\eeq
for any $\psi_i \in \Gamma(M,E)$. When $\rho$, considered as a section
of $Hom (E,TM)$, is nonzero at any point of
$M$, remarkably the invariance can be also concluded from the
remaining three axioms of a Courant algebroid cf.~\cite{Hansen-Strobl}. Also,
it is easy to see that just
Eq.~(\ref{ad}) by itself permits to conclude the Leibniz identity
(\ref{Leibniz}): Replace $\psi_2$ by $f \psi_2$ to obtain
$([\psi_1,f\psi_2],\psi_3) - f ([\psi_1,\psi_2],\psi_3) =
(\rho(\psi_1)f) (\psi_2,\psi_3)$, which yields the claimed equation by
bilinearity and non-degeneracy of the inner product. So, any anchored
vector bundle equipped with some bracket on its sections and an invariant
fiber metric in the sense of
(\ref{ad}) is an algebroid as defined in the introduction.

Sometimes in the literature an antisymmetrization of the above bracket
$[ \cdot , \cdot]$ is used in the definition, cf.~e.g.~the
contribution of Hitchin to this volume;\footnote{At least in the case
of exact Courant algebroids with $H=0$, introduced below, the
antisymmetrized bracket is called the Courant bracket and the bracket
used above the Dorfmann bracket.  In the axiomatization of a general
Courant algebroid, both types of brackets can be used, and, following
\cite{SeveraWeinstein}, we prefer the non-antisymmetric bracket for
the reason to follow. Although in the more recent literature
the name
Dorfmann bracket seems to prevail for the nonantisymmetric bracket in an
exact Courant algebroid, we will often use the terminus ``Courant
bracket'' for it.}  this antisymmetrized bracket, however, has
\emph{neither} of the two nice properties (\ref{Leibniz}) and
(\ref{Loday}), for which reason we preferred the present version of
axiomatization.

For $M$ a point, the definition of a Courant algebroid is
easily seen to reduce to a quadratic Lie algebra, i.e.~to a
Lie algebra
endowed with a nondegenerate, ad-invariant inner product:
Indeed, since then the r.h.s.~of the last equation in definition
\ref{Cour} vanishes, the bracket becomes antisymmetric and (\ref{ad})
apparently reduces to the notion of ad-invariance of the metric.
(Evidently in this context, or whenever $\rho$ has zeros, the
invariance needs to be required separately).  The fact that Courant
algebroids provide a generalization of quadratic Lie algebras may be
one motivation for introducing them, since in particular there is also
the following simple observation:
\begin{prop} Let $E$ be a Lie algebroid. Then it admits an invariant
fiber metric  $\Eg$ only in the (very restrictive) case of
$\rho \equiv 0$.
\end{prop}
\begin{proof} The statement follows at once from the formula (deduced
from the Leibniz property of the $E$-Lie derivative and formula
(\ref{Lie}) applied to $E$-1-forms)
\beq \EL_{f \psi} \Eg = f  \EL_{ \psi} \Eg + 2 \Emd f \vee
\iota_\psi \Eg \, ,
\eeq
valid for any $f \in C^\infty(M)$, $\psi \in \Gamma(M,E)$,
where $\vee$ denotes the symmetrized tensor product: according to
the first formula (\ref{diff-can}),
$\Emd f$ vanishes for all functions $f$ only if $\rho \equiv 0$.
$\square$ \end{proof}

Nevertheless, the option (\ref{ad}) is not the only possibility
to generalize quadratic Lie algebras to the realm of algebroids.
There are at least two more, which which will also play a role
in an attempt to generalize Yang-Mills theories to the context of structural algebroids/groupoids, as explained in the last section.
We thus introduce these two, which are also interesting in their own right.

Using the notion of $E$-Lie derivatives and inspired by the
philosophy to treat $E$ as the tangent bundle of a manifold with
the ensuing geometrical intuition, we consider
\begin{deff} \label{homo} Let $E$ be a Loday algebroid. We call it a
\emph{maximally symmetric $E$-Riemannian space} if $E$ is
equipped with a definite fiber metric $\Eg$ which permits a
(possibly overcomplete) basis of
sections $\psi_\alpha \in \GE$, $\langle \psi_\alpha(x) \rangle = E_x
\; \forall x \in M$, $\alpha = 1 \ldots s \ge \rank E$,
such that $\EL_{\psi_\alpha} \Eg =0$.
\end{deff}
Likewise we can define a maximally symmetric $E$-pseudo-Riemannian space by
dropping the condition of the definiteness of the fiber metric.
For $E=TM$ this reduces to the standard notion of a maximally symmetric
pseudo-Riemannian space (manifold), whereas for a Lie algebra
(a Lie algebroid over a point) this reproduces the notion
of quadratic Lie algebras.

So, in contrast to (\ref{ad}), \emph{this} notion is compatible with a
Lie algebroid with nonvanishing anchor, although it also poses
restrictions (like it is already the case in the example of Lie
algebras). Another option for generalizing quadratic Lie algebras to
Lie algebroids is the following one using an algebroid type covariant
derivative. For this we define
\begin{deff} Let $E \to M$ be an algebroid
and $V\to M$ another vector bundle over the same base manifold. Then
an \emph{$E$-covariant derivative} ${}^E{}\nabla$ on $V$ is a map from
$\Gamma(E) \otimes \Gamma(V) \to \Gamma(V)$, $(\psi,v) \mapsto
{}^E{}\nabla_\psi v$ that is $C^\infty(M)$ linear in the first entry
and satisfies the Leibniz rule \beq {}^E{}\nabla_\psi (fv) = f \,
{}^E{}\nabla_\psi v + \rho(\psi)f \, v \qquad \forall f \in
C^\infty(M) \, .\eeq
\end{deff}
Clearly, this notion reduces to an ordinary covariant derivative on
$V$ for the case of the standard Lie algebroid $E=TM$. Moreover, any
ordinary connection $\nabla$ on $V$ gives rise to an $E$-covariant
derivative by means of ${}^E{}\nabla_\psi := \nabla_{\rho(\psi)}$,
while certainly not every E-connection is of this form
(for $\rho
\equiv 0$, e.g., this expression vanishes identically, while in this case
\emph{any} section of $E^* \otimes End(V)$ defines an $E$-connection). In the
particular case of a Lie algebroid (and only there) one can generalize
the notions of curvature and torsion to the algebroid setting, with
${}^ER \in \Omega_E^2 \otimes \Gamma(\End V)$ and ${}^ET \in
\Omega_E^2 \otimes \Gamma(E)$. Note that a ``flat'' $E$-connection,
i.e.~one with ${}^ER =0$, is what one calls a \emph{Lie algebroid
representation} on $V$ (generalizing the ordinary notion of a Lie
algebra representation on a vector space, to which this reduces for
the case of $M$ being a point).

Given a fiber-metric on $V$ it is called compatible with an
$E$-connection on $V$, if it is annihilated by $E$-covariant
derivatives. In generalization of the well-known fact of the
uniqueness of a metric-compatible, torsion-free connection on a
manifold ($E=TM$), also for a general Lie algebroid $E$ one can prove
that there is a unique $E$-torsion-free $E$-connection compatible with
a (pseudo-)Riemannian metric ${}^Eg$. On the other hand, given an
ordinary connection $\nabla$ on an algebroid $E$, we can also regard
the following canonically induced $E$-connection
\beq \widetilde{{}^E{}\nabla}_{\psi_1} \psi_2 = \nabla_{\rho(\psi_2)} \psi_1 + [\psi_1,\psi_2] \, , \label{Enabla}
\eeq
which, for $E$ a Lie algebroid, differs from  ${}^E{}\nabla_\psi := \nabla_{\rho(\psi)}$ by subtraction of its own $E$-torsion.
It is easy to see that a Lie algebroid $E$ with fiber metric ${}^Eg$ compatible with $\widetilde{{}^E{}\nabla}$, induced by some ordinary connection on $E$, reduces to a quadratic Lie algebra; here one may or may not want to impose that this $E$-connection is flat, while only in the first case one would call it an $E$-metric invariant under the adjoint \emph{representation} induced by the connection $\nabla$.

Before discussing the second motivation of introducing Courant algebroids, which will lead us into the world of Dirac  and generalized complex structures, we use the opportunity of having introduced $E$--covariant derivatives for a
natural generalization of Courant algebroids:\footnote{This
notion was introduced independently and with different, but
complementary motivations in \cite{GruetzmannStrobl} and
\cite{Chinesen}.}
\begin{deff} A vector bundle twisted or \emph{$V$-twisted Courant algebroid}
is given by a Loday algebroid structure $(E,\rho,[\cdot, \cdot])$ on $E \to M$ together with a second vector bundle $V \to M$, an $E$-covariant derivative ${}^E{}\nabla$ on $V$,
and a surjective, non-degenerate bilinear map $( \cdot, \cdot ) \colon E \times_M E \to V$ such that
\beq ([\psi,\psi],\psi') = \frac{1}{2} {}^E{}\nabla_{\psi'} \, (\psi,\psi) = ([\psi',\psi],\psi) \, .
\eeq
\end{deff}
It is not difficult to see that for the case of $V$ being a trivial $\R$-bundle over $M$, admitting an $E$--covariantly constant basis section, the above definition reduces to the one of an ordinary Courant algebroid. A $V$-twisted Courant algebroid defined over a point does no more reduce to Lie algebras, but rather gives a Leibniz algebra $E$ only, with $V$ being an $E$-module  (cf.~\cite{GruetzmannStrobl} for further details on these statements).
$V$-twisted Courant algebroids play a role in the context of higher gauge theories, in particular nonabelian gerbes.

For later use we also recall the definition of a (strict) Lie 2-algebroid from \cite{GruetzmannStrobl}:
\begin{deff} \label{Lie2}
A \emph{strict Lie 2-algebroid} are two Lie algebroids over $M$,
$(E\to M,\rho,[ \cdot , \cdot ])$ as well as $(V\to M,0,[ \cdot , \cdot ]_V)$ together with a morphism $t\colon V \to E$ and a representation $\Ena$ of
$E$ on $V$ such that \beq \Ena_{t(v)} w= [v,w]_V \quad , \qquad t(\Ena_{\!
\psi} v) = [\psi, t(v)] \, , \qquad \forall v,w \in \Gamma(V), \forall \psi \in \Gamma(E) \, .\label{repr2} \eeq
\end{deff}
Note that the second Lie algebroid $V$ is a bundle of Lie algebras only, $\rho_V \equiv 0$. The general notion of a morphism for Lie algebroids is somewhat involved but most easily defined in terms of the Q-language developped later (and thus provided in Sec.~5 below); over the same base manifold as here, however, it simply implies $t([v,w])=[t(v),t(w)]$ as well as $\rho \circ t = \rho_V$.
For $M$ being a point, the above definition reduces to the one of a strict Lie 2-algebra or, equivalently, a differential crossed module. The above definition may be twisted by an $E$-3-form $H$ taking values in $V$, cf.~Theorem 3.1 in \cite{GruetzmannStrobl}, in which case one obtains what one might call an $H$-twisted or semistrict Lie 2-algebroid or simply a Lie 2-algebroid. (As one of the names suggests, it reduces to a semistrict Lie 2-algebra when $M$ is a point. The last name is most natural in view of the relation to so-called Q-manifolds, cf.~section 5 below).

\vskip3mm

We now return to ordinary Courant algebroids. The other motivation for considering them comes from the quest for
generalizing the notions of
(pre-) symplectic, Poisson, and complex manifolds to
socalled (real and/or complex) Dirac structures, our subject
in the next section. For their definition one
restricts to
particular kinds of Courant algebroids, so called exact Courant
algebroids: For a general Courant algebroid one obviously has
the following sequence of vector bundles
 \beq 0\rt T^*M
\stackrel{\rho^*}{\longrightarrow} E
\stackrel{\rho}{\longrightarrow} TM \to 0 , \label{exact} \eeq
where $\rho^*$ is the fiberwise transpose of $\rho$ combined with
the isomorphism induced by $\Eg$.\footnote{In fact, this even
defines a
complex of sheaves, since for a general Courant algebroid
$\rho \circ \rho^* \equiv 0$, cf.~e.g.~\cite{Courant}.}
\begin{deff} An exact Courant algebroid is a Courant algebroid such
that the sequence (\ref{exact}) is exact. \label{def:exact}
\end{deff}
For the rest of this section $E$ will always denote an
exact Courant algebroid.
\begin{proposition}\label{exact-Courant}
The image of $T^*M$ in $E$ is a maximally isotropic subbundle
w.r.t.~$\Eg$.
\end{proposition}
\begin{proof} By the definition of $\rho^*$, $(\rho^* \omega ,
\psi)=\langle \omega ,\rho
(\psi)\rangle$ therefore $(\omega_1 ,\omega_2)=0$ for all
$\omega_i\in \Omega^1 (M)$. Thus  $\rho^*T^*M
 \subset (\rho^*T^*M)^\perp$; by a dimensional argument, using that
the $E$-metric is nondegenerate, we may conclude equality.
\end{proof}
\noindent
\begin{deff}\label{def:Courant-connection} A splitting (also sometimes called a connection) of an exact
Courant algebroid is a map $j \colon TM \to E$, such that $\rho \circ j
= \mathrm{id}$ and $j(TM)$ is isotropic.
\end{deff}
Evidently $j(TM)$ defines a maximally isotropic subbundle of $E$
complementary to $T^*M$.  Hereafter (within this and the following two
sections) we identify a 1-form with its image in the exact Courant
algebroid and, given a splitting, we can likewise do so for the
image of vector fields w.r.t.~the splitting. A general element of $E$
can thus be written as the sum of an element of $T^*M$ and $TM$,
provided a splitting is given. Existence of the latter one is
guaranteed by the following
\begin{lemma} There always exists a splitting in $E$. The set of
splittings is a torsor over $\Omega^2 (M)$.
\end{lemma}
\begin{proof}
Let us take any splitting of the exact sequence $0\to T^*M\rt E\rt
TM\to 0$ in the category of vector bundles, $j : TM\to E$. Since
$T^*M$ is maximally isotropic, the pairing between $j (TM)$ and
$T^*M$ has to be nondegenerate. Therefore for any vector field
$\xi$ there exists a unique 1-form $\beta (\xi )$ such that
$(j (\xi), j(\xi'))=\langle \beta (\xi), \xi'\rangle $. The new
splitting
given by $j (\xi)-\frac{1}{2}\beta (\xi)$ is
maximally isotropic (note that for any 1-form $\alpha = \beta(\xi)$
one has by definition of $\rho^*$ and $j$: $(\rho^*(\alpha),j(\xi')) = \langle
\alpha, (\rho \circ j) ( \xi')\rangle = \langle \alpha ,
\xi' \rangle$). Any other splitting differs from the chosen
one by a section $B$ of $T^*M\otimes T^* M$, i.e. by sending
\beq\label{change-of-splitting}j_B (\xi)= j (\xi )+ B(\xi
,\cdot)\;.\eeq The tensor field $B$ is necessarily skew-symmetric
because the image of $TM$ is required to be maximally isotropic.
$\square$
\end{proof}

\begin{prop} For each $ \omega\in\Omega^1
(M)$ and $\psi\in\Gamma (E)$ one has:
\beq\label{bracket-forms} [\psi ,\rho^*\omega
]=\rho^*L_{\rho (\psi)}\;\omega \;, \quad [\rho^*\omega , \psi]
=-\rho^*\i_{\rho (\psi)} \md\omega\;. \eeq
\end{prop}
\begin{proof}
Invariance of the $E$-metric implies \beqn \rho (\psi_1) (\rho^*
\omega,\psi_2)=([\psi_1, \rho^* \omega ] ,\psi_2) +(\rho^* \omega
,[\psi_1,\psi_2])\;,\eeq hence $([\psi_1, \rho^* \omega ]
,\psi_2)=\langle L_{\rho (\psi_1)}\;\omega, \rho (\psi_2)\rangle$,
where we used $ \rho \left([\psi_1,\psi_2] \right) =
[\rho (\psi_1),\rho(\psi_2)]$.  The
$E$-metric being nondegenerate, we may conclude the first rule of
multiplication. The second one follows from this by the symmetry
property of the bracket in a Courant algebroid: $[\psi ,\psi]=\rho^*
\md (\psi ,\psi)$. $\square$
\end{proof}

A direct consequence  is:
\begin{cor} \label{cor:exact-Courant}The image of $T^*M$ is an abelian ideal of $E$.
\end{cor}

A result in the classification of exact Courant algebroids,
ascribed to P.~Severa \cite{Severalett}, is the following
one:
\begin{prop}\label{class_Courant}
Up to isomorphism there is a one-to-one correspondence between exact
Courant algebroids and elements of $H^3(M,\R)$. In particular, with the choice
of a splitting of (\ref{exact}) it takes the form
\beq\label{form-of-Courant} [\xi +\a ,\xi' +\a']=[\xi ,\xi']+H(\xi
,\xi', \cdot )+L_\xi\a' -\i_{\xi'}\md\a \;,\eeq where
$H\in\Omega^3 (M)$ and $\md H=0$. For a change of splitting
parameterized by a two-form  $B$ as in
(\ref{change-of-splitting}), the $3-$form transforms as $H\mto
H+\md B$.
\end{prop}

\begin{proof} Given a splitting of (\ref{exact}), one has
$[j(\xi), j(\xi' )]=j([\xi ,\xi'])+ H(\xi ,\xi')$, with $H$ the
``curvature'' of the splitting, taking values in 1-forms. The Leibnitz
property of the Courant bracket applied to the formula above
implies that $H$ is $C^\infty(M)$-linear with respect to its
first two arguments, hence
$H$ can be identified with a section of $(T^*M)^{\otimes\;3}$. The
image of $TM$ by $j$ is an isotropic subbundle of $E$. Therefore the
Courant bracket restricted to $j (TM)$ is skew-symmetric and therefore
also the curvature $H$, i.e.~$H(\cdot,\cdot)$. Moreover, using that
the Courant metric is invariant w.r.t.~the Courant bracket and that
the image of $j$ is isotropic, we obtain \beqn
\langle H(\xi_1 ,\xi_2 ),\xi_3\rangle=([j(\xi_1 ) ,j(\xi_2)],j
(\xi_3) )=-(j(\xi_2 ), [j(\xi_1),j (\xi_3) ])=-\langle
H(\xi_3,\xi_2),\xi_1\rangle\;,\eeq thus $H$ is totally skew-symmetric
and can be regarded as a 3-form. By the formula (\ref{bracket-forms})
the bracket of two sections of $E$ with the splitting $j$ is indeed
found to take the form of (\ref{form-of-Courant}).

Let us calculate the Jacobiator of sections $J
(\psi_1, \psi_2, \psi_3)=[\psi_1,[\psi_2,
\psi_3]]-[[\psi_1,\psi_2], \psi_3]]-[\psi_2,[\psi_1, \psi_3]]$ for
$\psi_i\in\Gamma (E)$, which has to vanish by the property of a
Courant algebroid. Again we want to be explicit here.
Using $\psi_i=\xi_i+\a_i$, where $\xi_i$ are
vector fields and $\a_i$ 1-forms, we obtain \beqn
[\psi_1,[\psi_2, \psi_3]]=\left[ \xi_1+\a_1 ,[\xi_2,
\xi_3]+H(\xi_2,\xi_3,\cdot) +L_{\xi_2}\a_3
-\i_{\xi_3}\md\a_2\right]= [\xi_1,[\xi_2, \xi_3]]+
\\\nonumber + H(\xi_1, [\xi_2 ,\xi_3], \cdot )+L_{\xi_1}
H(\xi_2,\xi_3,\cdot)+L_{\xi_1}L_{\xi_2}\a_3-L_{\xi_1}\i_{\xi_3}\md\a_2
 -\i_{[\xi_2 ,\xi_3]}\md\a_1\;,\eeq
\beqn [[\psi_1,\psi_2], \psi_3]]=\left[ [\xi_1,
\xi_2]+H(\xi_1,\xi_2,\cdot) +L_{\xi_1}\a_2 -\i_{\xi_2}\md\a_1,
\xi_3+\a_3 \right]=[[\xi_1, \xi_2],\xi_3]+\\\nonumber +H([\xi_1,
\xi_2 ],\xi_3, \cdot )+L_{[\xi_1,\xi_2]}\a_3 -\i_{\xi_3}\md
H(\xi_1,\xi_2,\cdot)
-\i_{\xi_3}L_{\xi_1}\md\a_2+\i_{\xi_3}L_{\xi_2}\md\a_1\;,\eeq
\beqn [\psi_2,[\psi_1, \psi_3]]=\left[ \xi_2+\a_2 ,[\xi_1,
\xi_3]+H(\xi_1,\xi_3,\cdot) +L_{\xi_1}\a_3
-\i_{\xi_3}\md\a_1\right]= [\xi_2,[\xi_1, \xi_3]]+
\\\nonumber + H(\xi_2, [\xi_1 ,\xi_3], \cdot )+L_{\xi_2}
H(\xi_1,\xi_3,\cdot)+L_{\xi_2}L_{\xi_1}\a_3-L_{\xi_2}\i_{\xi_3}\md\a_1
 -\i_{[\xi_1 ,\xi_3]}\md\a_2\;,\eeq
so one can rewrite the Jacobiator as
\beqn J (\psi_1, \psi_2,
\psi_3)&=&[\xi_1,[\xi_2,\xi_3]]-[[\xi_1, \xi_2],\xi_3]-[\xi_2,[\xi_1,
\xi_3]]\\\nonumber &+& L_{\xi_1}L_{\xi_2}\a_3
-L_{[\xi_1,\xi_2]}\a_3-L_{\xi_2}L_{\xi_1}\a_3\\\nonumber
&+&-L_{\xi_1}\i_{\xi_3}\md\a_2+\i_{\xi_3}L_{\xi_1}\md\a_2
+\i_{[\xi_1 ,\xi_3]}\md\a_2\\\nonumber &+& -\i_{[\xi_2
,\xi_3]}\md\a_1 -
\i_{\xi_3}L_{\xi_2}\md\a_1+L_{\xi_2}\i_{\xi_3}\md\a_1\\\nonumber
&+& (\md H )(\xi_1
,\xi_2,\xi_3,\cdot).\eeq Apparently, the
first four lines vanish and  therefore the $3-$form $H$ entering into the
product formula has to be closed. Finally, by a simple calculation one
obtains \beqn [\;\psi_1 +B(\rho
(\psi_1),\cdot),\psi_2 +B(\rho (\psi_2),\cdot)\;]=[\psi_1,\psi_2]
+\md B(\rho(\psi_1),\rho (\psi_2),\cdot)\;,\eeq which finalizes
the proof of the proposition \ref{class_Courant}. $\square$
\end{proof}




Suppose a splitting of (\ref{exact}) is chosen. Then, each section
of $E$, $\psi =\xi +\a$, can be thought of as an operator acting
on differential forms by the formula
\beq\label{spin-action}\mathtt{c}(\psi)\,\omega =\i_\xi \,\omega
+\a\wedge\omega\,.\eeq It is useful to consider the space of forms as
a spinor module over the Clifford algebra of $E$ with the
quadratic form given by the Courant metric. The next simple
proposition, the proof of which we leave for readers, shows that
the notion of derived brackets can be exploited also in the case
of Courant algebroids. 
\begin{proposition}\label{derived} The following identity holds true:
\beq\label{Courant-derived}
\mathtt{c}([\psi_1,\psi_2])=\left[[\mathtt{c}(\psi_1), \md +H],\mathtt{c}(\psi_2)\right]\;.\eeq
\end{proposition}


\section{Dirac structures}
\label{sec:Dirac}

\begin{deff}
\label{def:Dirac} Suppose $E$ is an exact Courant algebroid over $M$
and $D$ a maximal totally isotropic (maximally isotropic) subbundle of $E$
with respect to the Courant scalar
product. Then $D$ is called a Dirac structure, if the Courant
bracket of two sections of $D$ is again a section of $D$.
\end{deff}

\begin{example} $T^*M$ is a Dirac structure, see
Proposition \ref{exact-Courant} and Corollary \ref{cor:exact-Courant}.
Given a connection in $E$ (see Definition \ref{def:Courant-connection})
$TM$ is a maximally isotropic subbundle of $E$; $TM$ is a Dirac
structure, iff the curvature $H$ is zero. \label{ex:TM}
\end{example}

\noindent Suppose a splitting of a Courant algebroid $E$ is
chosen, $E=T^*M\oplus TM$, the curvature $H$ of which is zero. We then
have the following two examples showing that Poisson and presymplectic
geometry give rise to particular Dirac structures.

\begin{example} \label{ex1} Let $D$ be a graph of a
tensor field $\Pi \in \Gamma(\otimes^2TM)$, considered as a map from
$T^*M$ to $TM$. Then $D$ is a Lagrangian subbundle iff $\Pi$ is
skew-symmetric, i.e.~iff it is a bivector field, $\Pi \in \Gamma
(\Lambda^2 TM)$. The projection of $D$ to $T^*M$ is non-degenerate.
Vice versa, any Lagrangian subbundle of $E$ which has a
non-degenerate projection to $T^*M$ is a graph of a bivector
field. The graph of a bivector field $D$ is a Dirac structure if and
only if $\Pi$ is a Poisson bivector, $[\Pi, \Pi]=0$.
\end{example}
\begin{example} \label{ex2} Let $D$ be a graph of a
$2-$form $\omega\in \Omega^2 (M)$ considered as a skew-symmetric
map from $TM$ to $T^*M$, then $D$ is a Lagrangian subbundle and
the projection of $D$  to $TM$ is non-degenerate. Any Lagrangian
subbundle of $E$ which has a non-degenerate projection to
  $TM$ is a graph of a $2-$form. $D$ is a Dirac structure
 if and only if $\omega$ is closed.
\end{example}
In the second example we could also have started with a general (2,0)-tensor-field
certainly, completely parallel to the first example.  For the case
that $H$ is non-zero, one obtains twisted versions of both
examples, cf.~Eq.~(\ref{HPoisson}) in the case of example \ref{ex1}.

The complexification of $E$, $E^c=E\otimes_\R \C$, is a real bundle,
i.e. a bundle over $\C$ endowed with a $\C$--anti-linear operator
$\sigma$ acting on $E^c$, such that $\sigma^2 =1$ and $E=\ker(\sigma
-1)$.
We next turn to
a description of the algebraic set of complex subbundles of $E^c$, which are
maximally isotropic with respect to the complexified Courant scalar
product. In order to do so, we introduce an additional structure,
namely a particular kind of fiber metric $\Eg' =: g_\tau$ (cf.~also
\cite{KSS}), different from the canonical one $\Eg$,
called sometimes a generalized Riemannian metric. Note
that according to the philosophy presented here, a generalized Riemannian metric on
a Courant algebroid would be any positive definite fiber metric;
we will still adopt this partially established terminology now
and then show that, given a splitting, it
corresponds to an ordinary Riemannian metric $g$ on $M$ together with a
2-form $B$ (cf.~\cite{KSS} as well as the contribution of N.~Hitchin within this volume).

\begin{deff}\label{generalized-Riemann}
A positive $E$-metric  $g_\tau$ which can be expressed
via an operator $\tau \in \Gamma(\End(E))$ squaring to the identity
such that for any $\psi_i\in\Gamma (E) $ \beqn g_\tau (\psi_1 ,\psi_2) &=& (\tau
\psi_1, \psi_2)\eeq  is called a
generalized Riemannian metric. 
\end{deff}

\noindent It follows from the definition, that $\tau$ has to be
self-adjoint with respect to the Courant metric
$\Eg(\cdot,\cdot)\equiv(\cdot,\cdot)$; because of
$\tau^2=1$, it is orthogonal moreover. From $g_\tau (\psi,\psi) \ge 0$
we conclude that the $+1$
and $-1$ eigen-subspaces of $\tau$ are positive and negative
definite, respectively,  again with respect to the canonical metric in $E$, while
$(\tau \psi_1,\tau \psi_2)= (\psi_1,
\psi_2)$ shows that they are orthogonal to one another. As a corollary, $\Eg$ having signature $(n,n)$, the
dimensions of the two eigenvalue subspaces are equal. Vice versa, let
us take any positive definite subbundle $V$ of maximal rank, then
there exists a unique $\tau$ which is postulated to be $1$ on $V$ and
$-1$ on the orthogonal subbundle $V^\perp$ (with respect to the
canonical metric). The operator $\tau$ satisfies the properties of the
definition as above.

Given a splitting of the exact sequence
(\ref{exact}), there is a one-to-one correspondence between
generalized Riemannian metrics in the sense of the Definition
\ref{generalized-Riemann} and the one of Hitchin within this volume.
Since $V$ is positive, it has zero intersection with any Lagrangian
subspace of $E$ and in particular with $TM$ and $T^*M$ (i.e.~with
$j(TM) \cong TM$ and $\rho^*(T^*M) \cong T^*M$).  Hence $V$
is a graph of an invertible bundle map $TM\to
T^*M $ which can be identified with a non-degenerate tensor in $\Gamma
(T^*M\otimes T^*M)$ such that its symmetric component $g$ in the
decomposition into the symmetric and antisymmetric parts, $g+B$, is a
Riemannian metric on $TM$.

Define a new real form by using the
complexification of $\tau$: $\tau_\sigma :=\sigma \tau$.  Since $\tau$
is a real operator, its complexification commutes with the complex
conjugation $\sigma$, therefore $\tau_\sigma^2=1$. Moreover,
$\tau_\sigma$ is an anti-linear operator as a composition of linear
and anti-linear operators, thus it defines a new real structure in
$E^c$ (different from the old one $\sigma$).
Hereafter we shall distinguish the $\sigma$- and $\tau_\sigma$-real forms
by calling them "real" and "$\tau$-real", respectively.

\begin{proposition}
The real and imaginary subbundles of $E^c$ with respect to the real
structure $\tau_\sigma$ are the subbundles (over $\R$)
$E^+:=V\oplus iV^\perp$ and $E^-:=V^\perp\oplus iV$, respectively.
Any complex Dirac structure in $E^c$ is a totally complex subbundle
w.r.t.~$\tau_\sigma$, i.e. its intersections with the
$\tau_\sigma$-real and totally complex subbundles of $E^c$ are trivial.
\end{proposition}
\proof The first statement follows trivially from the definition
of $\tau_\sigma$: one needs to take into account that
$\tau_\sigma$ is anti-linear. To check the second statement, it
suffices to notice that the (real) eigen-subspaces of the new real
structure $\tau_\sigma$, $E^+$ and $E^-$, are positive (negative) definite
with respect to the restriction of the (complexified) Courant metric
in $E^c$, therefore $E^\pm$ do not contain any nonzero isotropic vectors.
$\square$

Let $D$ be a maximally isotropic complex subbundle of $E^c$; by
the proposition above, $D$ is a totally complex subbundle as well as
$\tau_\sigma (D)$. Define a linear complex structure $J$ in $E$ by
requiring $J (D)=i$ and $J (\tau_\sigma (D))=-i$. By construction,
$J$ is a $\tau-$real operator. Moreover, the restriction of $J$ to
$E^+$ is an orthogonal operator with respect to the induced
positive metric in $E^+$.

\begin{proposition}
 There is a one-to-one correspondence between complex maximally
 isotropic subbundles in $E^c$ and $\R-$linear orthogonal operators
 $J$ in $E^+$ which satisfy $J^2=-1$. In particular, if $D$ is a
 real maximally isotropic subbundle,
 then $J$ is uniquely represented by a real orthogonal operator
 $S: V\to iV^\perp$ as follows:
 \beq\label{J-vs-S}
 J_\sigma :\left( \be{cc} 0 & -S^{-1}\\ S & 0\ee\right)\;.
 \eeq
\end{proposition}
\proof We need to check only the last statement. Suppose $D$ is a
complexification of some real maximally isotropic subbundle of
$E$, then $D$ is preserved by the real structure $\sigma$ and
$\tau_\sigma (D)=\tau (D)$. Now, by construction, $J$ anti-commutes
with $\sigma$. Let us identify $E^+$ with $V\oplus V^\perp$ such
that $\sigma$ acts as $+1$ on the first and as $-1$ on the second
factors. This identification is an isometry, if we supply $V^\perp
$ with an opposite metric (the original one on $V^\perp$  is
negatively definite). Since $J$ anti-commutes with $\sigma$, it
has only off-diagonal entries. Now, taking into account that $J$
is orthogonal and squares to $-1$, we immediately get the required
form (\ref{J-vs-S}). $\square$

\begin{proposition} Let $D$ be a maximally isotropic real subbundle of $E$
provided with the metric induced by $g_\tau$. Then the projector of $D$
to $V$ and $V^\perp$ is an isometry up to a factor $\frac{1}{2}$ and
$-\frac{1}{2}$, respectively.
\end{proposition}
\proof Taking into account that the orthogonal projector to $V$
and $V^\perp$ can be expressed as $P_V=\frac{1}{2}(1+\tau)$ and
$P_V^\perp=\frac{1}{2}(1-\tau)$, respectively, we obtain: \beqn
(P_V\psi_1, P_V\psi_2)&=& \frac{1}{2}
\left((\tau\psi_1,\psi_2)+(\psi_1,\psi_2)\right)\;\\\nonumber
(P_V^\perp\psi_1, P_V^\perp\psi_2)&=& \frac{1}{2}
\left(-(\tau\psi_1,\psi_2)+(\psi_1,\psi_2)\right)\;\eeq for all
$\psi_1,\psi_2\in\Gamma (E)$. Imposing $\psi_i\in\Gamma (D)$, we
 get the required property. $\square$

\vskip 0.2em \noindent Let us take $D=T^*M$ supplied with the
positive metric, induced by $g_\tau$ in $E$ as above.
\begin{cor}
There is a one-to-one correspondence between real maximally
isotropic subbundles and orthogonal operators acting point-wisely
in $TM$. \label{oneone}
\end{cor}
\proof The composition of maps $P_V^\perp\circ ({P_V}_{\mid D})^{-1}: V\to
V^\perp$ is an isometry up to a factor $-1$ (an anti-isometry). Combining
this map with the natural anti-isometry of real spaces,
$V^\perp\to iV^\perp $, given by the multiplication with $i$, we get
an isometric identification of $V$ and $iV^\perp$, which allows to
identify $S$ with an orthogonal operator acting in $V$. Using the
conjugation by $P_V$, we identify $S$ with an orthogonal operator
acting in $T^*M$ or, by the duality, in $TM$. $\square$

\vskip 0.2em The above operator $S$ is a section of $O(TM)$, the
associated bundle $P\times_{O(n)} O(n)$ where $P$ is the bundle of
orthogonal frames and $O(n)$ acts on itself by conjugation.  As it
was constructed, the homotopy class of the section $S$ depends
only on the homotopy class of the corresponding maximally
isotropic subbundle. One knows that, for a Lie group $G$, any
principal $G-$ bundle over $M$ is the pull-back by the canonical
one over the universal classifying space, $EG\to BG$ with respect
to some map $\varphi_0 :M\to BG$. Then $P\times_G Y$ is the
pull-back of $EG\times_G Y$ for each $G-$ space $Y$. Applying this to
$G=Y=O(n)$, one has the following commutative diagram:
\beq\be{ccc} O(TM)
&\stackrel{\varphi }{\longrightarrow}& EO(n)\times_{O(n)} O(n)\\
\downarrow && \downarrow \\
M &\stackrel{\varphi_0}{\longrightarrow}& BO(n)\ee\;. \eeq Given
any section $S: M\to O(TM)$, the pull-back of $\varphi \circ S$
defines a map\footnote{This construction was
suggested by A. Alekseev during our joint work on Dirac
structures. } from the equivariant cohomology of $O(n)$ with
coefficients in some ring $\mathrm{\CR}$ to the cohomology of $M$ : \beq
(\varphi \circ S)^* : H^i (EO(n)\times_{O(n)}O(n), \mathrm{\CR})=H^i
(O(n), \mathrm{\CR})^{O(n)}\to H^i (M, \mathrm{\CR})\;.\eeq Since all
the arrows are defined up to homotopy, we obtain a
characteristic map from the product of the set of homotopy classes of maximally
isotropic subbundles with $H^i(O(n),\CR)$ to the cohomology of $M$.
The explicit
construction of this characteristic map for $\mathrm{\CR}=\R$ can be
done by use of the secondary characteristic calculus. Let $\na$ be
the Levi-Civita connection of the induced metric, $\varPhi$ an
ad-invariant polynomial on the Lie algebra $\mathfrak{o}(n)$, then
define the corresponding characteristic class as follows: \beq
\mathrm{c}_\varPhi (S,g) &=& \int\limits_0^1 \varPhi (\md t\wedge
A + F(t))\;,\eeq where $A=S^{-1}\na (S)$, and $F(t)$ is the
curvature of $\na +tA$. One can easily check that \beqn \md
\mathrm{c}_\varPhi (S,g) &=& \varPhi \left((S^{-1}\circ\na\circ
S)^2\right)-\varPhi (R)=0\;,\eeq where $R$ is the curvature of
$\na$ (the Riemann curvature), since $\varPhi$ is ad-invariant.
The cohomology class of $c_\varPhi (S,g)$ does not change under the
homotopy of $S$ and $g$.

\begin{proposition} Let $D$ be a maximally isotropic subbundle
of $E$. Suppose that there is a connection, i.e. an isotropic
splitting of $0\to T^*M\to E\to TM\to 0$, such that the projection
of $D$ to $TM$ or $T^*M$ is non-degenerate. Then the section $S$,
constructed as above, is homotopic to $1$ and $-1$, respectively,
thus all characteristic classes vanish.
\end{proposition}
\proof Obviously, whatever metric $g$ is taken, $T^*M$ corresponds
to  $-1\in \Gamma (O(TM))$. Assume that the
projection of $D$ to $T^*M$ is non-degenerate, then, apparently,
$D$ is homotopic to $T^*M$ and thus the section $S$, corresponding
to $D$, is homotopic to $1$. If the projection of $L$ to $TM$ is
non-degenerate, then $D$ is homotopic to image $TM$ and thus to
any maximally isotropic subbundle with zero intersection with
$T^*M$. Now it suffices to take the orthogonal complement of
$T^*M$ in $E$ with respect to the generalized Riemann metric $g_\tau$:
this is a maximally isotropic subbundle of the required type which
corresponds to $1$. $\square$

\begin{proposition} Let $D$ and $D'$ be maximally isotropic
subbundles corresponding to $S$ and $S'\in
\Gamma (O(T^*M))$, respectively, then the intersection
$D\cap D'$ is zero if and only if $S^{-1}S'-1$ is
non-degenerate.
\end{proposition}
\proof The proof follows from the explicit parameterization of $D$ by $T^*M$:
\beq D_x &=&\{(1+\tau)\eta +(1-\tau)S\eta
\;|\; \eta\in T^*_xM \}\;.\eeq
 $\square$

 \begin{cor} $D$ and $T^*M$ admit a common complementary maximal
 isotropic subbundle, if and only if there exists $R\in\Gamma
 (O(TM))$ such that $R-1$ and $RS-1$ are non-degenerate.
 \end{cor}

Previously we proved that if a Lagrangian subbundle $D$ and $T^*M$
has a common complementary maximally isotropic subbundle, then the
operator $S$ which corresponds to $D$ can be deformed to the minus
identity section (we take the connection in $E$ such that the
image of $TM$ coincides with the chosen common complementary
subbundle). In general, we can not reverse this statement, i.e.
even if $D$ is homotopic to $T^*M^\perp$, there may be no
splitting in $E$  such that the projection of $D$ to $T^*M$ becomes
non-degenerate.
\begin{example}\label{counterexample1}
Let $M$ be a $2-$dimensional surface, then $\det S=\det R=1$.
Therefore both operators $S$ and $R$ can be thought of as
points in $S^1$.

\vskip 2mm\noindent Let us assume that $M$ contains a loop such that $S$ passes $1\in S^1$
and covers $S^1$ many times in both directions. We only claim that
the total degree of the map from the loop to $S^1$ equals to zero
(hence $S$ is homotopically equivalent to the identity map). For
example, $M=T^2$. One can take a smooth function $\phi (t)=a\sin
(2\pi i t)$. The quotient $t\mto \phi (t) \mod \Z$  defines a
smooth map $S^1\to S^1$ of degree zero, because the limit of $\phi
(t)$ at $a\to 0$ is zero.
 Nevertheless,
the image of $\phi $ wraps $S^1$ many times for sufficiently
large $a$. Now we can trivially extend the map to $T^2$ which
gives $S$. Since the function, which would correspond to $R$, is
not permitted to coincide with $\phi$ and to pass through $1$, we
conclude that there is ``no room'' for $R$.
\end{example}

\section{Generalized complex structures}
\label{sec:GCS}


Generalized complex structures (GCS) were invented by Nigel Hitchin \cite{HitchinGCS}.
It turned out that GCS provides a mathematical background of certain sigma models (for instance, those
the target space of which is endowed with
a bihermitian metric, cf.~\cite{GatesHullRocek,Maxim}).
Generalized complex structures interpolate naturally between
 symplectic and holomorphic Poisson geometry.

\begin{deff}  Let $E$ be an exact Courant algebroid.
A generalized complex structure is a maximally isotropic pure totally complex
Dirac subbundle of the complexification of $E$ with the complexified Courant scalar product,
that is, a maximally isotropic subbundle $D$ of $E^c =E\otimes_\R \C$ such that $D$ is
closed w.r.t. the Courant bracket and $D\cap \overline{D}=\{0\}$.\footnote{We call such
a subbundle $D$ totally complex because its ``real part'' $D\cap\overline{D}$ is zero}
\end{deff}

\noindent  As in Section \ref{sec:Dirac},
we uniquely associate a maximally isotropic totally complex subbundle $D$ of $E^c$ with
 a real point-wisely acting operator $J$ which squares to $-1_E$,
such that, at any $x\in M$ the fibers $D_x$ and $\overline{D}_x$ are the $+i$
and $-i$ eigen-subspaces of $J_x$, respectively.
The isotropy condition of $D$ implies that $J$ is skew-symmetric with respect to
the Courant scalar product $(,)$.
From now on we shall consider only an ``untwisted'' version of an exact Courant algebroid
together with an isomorphism $E\simeq TM\oplus T^*M$, where the latter direct sum is supplied with
the canonical scalar product and the
(Dorfman) bracket (\ref{Courant-bracket}). We also
identify bivector fields and 2-forms on $M$ with sections of
$Hom (T^*M, TM)$ and $Hom (TM, T^*M)$, respectively, by use of the corresponding contractions.
 The next lemma gives a complete set of algebraic conditions for $J$.

\begin{lemma} $J$ is an operator of the form
\beq\label{form-of-J} J=\left(\be{cc} I & \Pi \\ \Omega & -I^* \ee
\right)\;,
\eeq
where $I\in \Gamma (End TM)$, $\Pi\in \Gamma (\Lambda^2 TM)$,
and $\Omega\in\Omega^2 (M)$. Then $J^2=-1_E$ translates into:
\beq\label{conditions-for-J}
\be{ccc}I^2 + \Pi\Omega &=&-1_{TM}\\ (I^*)^2 +\Omega\Pi &=&-1_{T^*M}\\
I^*\Omega -\Omega I&=& 0\\ I\Pi -\Pi I^* &=& 0\ee \;
\eeq
 \end{lemma}

\begin{proof}
Straightforward calculation.
$\square$
\end{proof}

\vskip 2mm

\noindent Similarly to the end of Section \ref{sec:GG}, we treat differential forms on $M$
as sections of the spin module over the Clifford bundle $C(E)$, that is, the bundle
of associative algebras generated by $E$ subject to the relations
$\varphi^2=(\varphi,\varphi)$ for all $x\in M$ and $\varphi\in E_x$. The spinor action of
$\Gamma( C(E))$ on $\Omega^*(M)$ is
the extension of  (\ref{spin-action}). Taking into account that the generalized complex
structure operator $J$ is skew-symmetric
 with respect to the scalar product, it can be thought of as a section of
$C(E)$, which we denote by the same letter $J$, such that
$J(\psi)=[J,\psi ]$ for each $\psi\in \Gamma (E)$. Here $[,]$ is the
(super-)commutator of sections of the Clifford bundle.  It is easy to check that
\beq
\mathtt{c}(J) (\a_1\wedge ...\wedge \a_p)= -\sum\limits_{i=1}^p  \a_1\wedge ...\wedge I^*(\a_i)
\wedge ...\wedge\a_p + (\i_\Pi +\Omega\wedge)\a_1\wedge ...\wedge \a_p\;,
\eeq
where $\a_i\in \Omega^1 (M)$, $i=1,\ldots, p$.

\vskip 2mm

\noindent We now
focus on the question when a maximally isotropic totally complex subbundle $D$ is closed w.r.t.~the
Courant bracket or, in other words, when it is integrable.

\begin{proposition} \label{prop:Nij} $D$ is a generalized complex structure if and only if
 for any $\psi,
\psi'\in\Gamma(E)$  one has
\beq\label{Nijenhuis-zero}
J\left( [J\psi , \psi']+[\psi, J\psi']\right)+
[\psi ,\psi']-[J\psi ,J\psi']=0\;.
\eeq
\end{proposition}
\begin{proof} The l.h.s.~of the equation (\ref{Nijenhuis-zero}) is bilinear and real, thus it is sufficient
to check the property for $\psi_1,\bar\psi_2$ and $\psi_1,\psi_2$, where $\psi_i\in \Gamma(D)$.
While in the first case the expression vanishes identically, the second one gives
\beqn J\left( [J\psi_1 , \psi_2]+[\psi_1 J\psi_2]\right)+
[\psi_1 ,\psi_2]-[J\psi_1 ,J\psi_2]=2i\left( J[\psi_1,\psi_2]-i[\psi_2,\psi_2]\right)\,,
\eeq
 which is identically zero if and only if $[\psi_1,\psi_2]$ is again a section of $D$. The latter
 is nothing else but the integrability condition for $D$.

$\square$\end{proof}

\noindent The meaning of equation (\ref{Nijenhuis-zero}) is the same
as for an almost complex structure in the usual sense, that is,
the vanishing of a certain tensor called the  (generalized)
 Nijenhuis tensor of $J$. The next Lemma provides an explicit
construction of the  Nijenhuis tensor in terms of the spin module.

\begin{lemma} The operator $N_J\colon =\frac{1}{2}\left(\left[\mathtt{c}(J),
[\mathtt{c}(J),\md]\right]+\md\right)$
is a point-wisely acting map, such that for any $\psi,
\psi'\in\Gamma(E)$ the following identity holds:
\beqn [\mathtt{c}(\psi'),[\mathtt{c}(\psi), N_J]]= \mathtt{c}\left(J\left( [J\psi , \psi']+
[\psi, J\psi']\right)+
[\psi ,\psi']-[J\psi ,J\psi']\right)\;.
\eeq
\end{lemma}
\begin{proof} The first part of the Lemma follows from the identity
\beqn
2[N_J, f]= \left[\left[\mathtt{c}(J), [\mathtt{c}(J),\md]\right],f\right]+
[\md ,f]=\mathtt{c}(J^2 (\md f)+\md f)=0\,,
\eeq
which holds for each $f\in C^\infty (M)$. The second part requires a
simple computation which essentially based upon the derived formula
of the Courant bracket (\ref{Courant-derived}) when $H=0$.
Indeed, let us write the Nijenhuis tensor
in the
form $N_J=\frac{1}{2}(ad_J^2 (\md)+\md)\,$ where $ad_J (a)\colon =[\gamma_J, a]$ for any
$a$. It is clear that
\beqn [\mathtt{c}(\psi), ad_J^2 (a)]=-[\mathtt{c}(\psi), a]-2ad_J \left([\mathtt{c}(J\psi), a]\right)
+ad_J^2 \left([\mathtt{c}(\psi), a]\right)\,.
\eeq
Thus, taking into account that $[ad_J, \mathtt{c}(\psi)]=\mathtt{c}(J\psi)$ and $J^2\psi=-\psi$,
we obtain \beqn [\mathtt{c}(\psi), ad_J^2 (\md)+\md]=-2ad_J \left(L_{J\psi}\right)
+ad_J^2 \left(L_\psi\right)\,,\eeq where $L_\psi\colon =[\md, \mathtt{c}(\psi)]$, and
finally
\beqn
[\mathtt{c}(\psi'), [\mathtt{c}(\psi), ad_J^2 (\md)+\md]]= -2[\mathtt{c}(\psi'),ad_J \left(L_{J\psi}\right)]
+[\mathtt{c}(\psi'), ad_J^2 \left(L_\psi\right)]=\\ \nonumber
2[\mathtt{c}(J\psi'), L_{J\psi}] -
2ad_J\left( [\mathtt{c}(\psi'), L_{J\psi}]\right) -
2ad_J\left( [\mathtt{c}(J\psi'), L_{\psi}]\right)
 -[\mathtt{c}(\psi'), L_\psi]
+ \\ \nonumber + ad_J^2 \left([\mathtt{c}(\psi'), L_\psi]\right)=2\mathtt{c}\left(J\left( [J\psi , \psi']+
[\psi, J\psi']\right)+
[\psi ,\psi']-[J\psi ,J\psi']\right)\;,
\eeq
which completes the proof.

\end{proof}

$\square$

\noindent As a corollary, the integrability condition (\ref{Nijenhuis-zero})
admits the following equivalent form:
\beq\label{Nijenhuis-zero2}
\left[\mathtt{c}(J),
[\mathtt{c}(J),\md]\right]+\md =0\,.
\eeq

\noindent Let us remark that, if we decompose the l.h.s. of (\ref{Nijenhuis-zero2})
into the sum of homogeneous
components with respect to the natural grading in $\Omega^* (M)$,
the lowest
degree term will give us $[\i_\Pi, [\i_\Pi,\md]=0$; this is equal to $\i_{[\Pi,\Pi]}=0$,
where $[\Pi,\Pi]$ is the Schouten-Nijenhuis bracket of $\Pi$. Thus the vanishing of $N_J$ implies,
in particular, that $\Pi$ is a Poisson bivector. The explicit derivation of the remaining
homogeneous terms gives a complete set of compatibility conditions
for $I$, $\Pi$, and $\omega$ (cf.~\cite{Crainic_brackets} for the details).

\begin{proposition} {\mbox{}\vskip 2mm}
 \begin{enumerate}
\item If $\omega =0$ then $I$ is a complex structure and $\Pi$
is a real part of a complex Poisson bivector which is holomorphic with respect $I$;
\item If, on the other hand, $I=0$ we reobtain
an ordinary symplectic structure from a generalized complex one, such that
 $\Pi$ then is the respective induced Poisson structure.
 \end{enumerate}
\end{proposition}
\begin{proof}

\vskip 1mm\noindent 1. The algebraic conditions (\ref{conditions-for-J}), imposed
on $J$,  will give us $I^2=-1_{TM}$ and $I\Pi=\Pi I^*$. The first identity
implies that $I$ is an almost complex structure in the usual sense, therefore the complexified
tangent and cotandent bundles admit the usual decomposition into
 the sum of $(1,0)$ and $(0,1)$ parts with respect to $I$:
\beqn
T(M)^c=T^{(1,0)}\oplus T^{(0,1)}\,,\hspace{3mm} T^*(M)^c={T^{(1,0)}}^*
\oplus {T^{(0,1)}}^*\,.
\eeq
The second commutation relation implies that
the complexification of $\Pi$ belongs to the direct sum of
$Hom ({T^{(1,0)}}^*, T^{(1,0)})$ and $Hom ({T^{(0,1)}}^*, T^{(0,1)})$.
Taking into
account that $\Pi$ is real, we get
$\Pi=\Pi_h+\overline{\Pi}_h$, where $\Pi_h\in\Gamma (\Lambda^{2}T^{(1,0)}M)$.
The degree $0$ homogeneous
component of the integrability condition (\ref{Nijenhuis-zero2})
asserts the vanishing of the Nijenhuis tensor of $I$, which means that $I$ is a complex structure
on $M$ in the usual sense. The degree $1$ homogeneous component of  (\ref{Nijenhuis-zero2})
gives us
\beq\label{compatibility_I_and_Poisson}
[I^*, \de_\Pi]+[\md_I, \i_\Pi]=0\,,
\eeq
where $\de_\Pi\colon =[\md , \Pi]$ and $\md_I \colon =[\md , I^*]$. The Hodge decomposition
of differential forms with respect to $I$ allows to decompose $\md$ into the sum
of $(1,0)$ and $(0,1)$ parts,  $\partial +\overline{\partial}$, and $\i_\Pi$ into
the sum of $(2,0)$ and $(0,2)$ parts (the contractions with
$\overline\Pi_h$ and $\Pi_h$, correspondingly). Similarly, we have the decomposition
of $\md_I$:
$\md_I =i(\partial -
\overline{\partial})$.
Let us look at the components of the l.h.s. of (\ref{compatibility_I_and_Poisson})
which are homogeneous
with respect to the Hodge decomposition. The identity  (\ref{compatibility_I_and_Poisson})
gives us only two (dependent) conditions, the first of which is conjugated to the second one:
$[\overline\partial, \i_{\Pi_h}]=0$ and $[\partial, \i_{\overline\Pi_h}]=0$.
This holds if and only if $\Pi_h$ is
a holomorphic Poisson bivector.

\vskip 2mm\noindent 2. Let $I=0$, then the algebraic relations (\ref{conditions-for-J})
gives us the only independent condition $\Pi\omega =-1_{TM}$, which means that
$\omega$ is a non-degenerate 2-form and $\Pi$ is
 the corresponding bivector field (up to
a sign convention). As it was mentioned above,
once $J$ is integrable, $\Pi$ is necessarily Poisson. Thus $\omega$
has to be a symplectic form. Conversely, provided $\omega$ is a symplectic form,
we define $\Pi$ such that $\Pi\omega=-1_{TM}$. Now we need to check
the integrability condition (\ref{Nijenhuis-zero2}). This gives
us the only identity to verify:
$[\de_\Pi, \omega\wedge]+\md=0$. Since $M$ is a symplectic manifold, its dimension is even,
say $\dim M =2m$. It is easy to check that
$[\i_\Pi, \omega\wedge] =\bar{n}-m$, where
$\bar{n}$ is an operator which counts the degree of a differential
form and $m$ is simply the multiplication on $m$.
The identity $\md\omega=0$ is obviously equivalent
to the commutator relation $[\md, \omega\wedge]=0$. Therefore
\beqn
[\de_\Pi, \omega\wedge]+\md =[[\i_\Pi,\md], \omega \wedge]=-[\md,[\i_\Pi, \omega\wedge]]+\md =
-[\md,\bar{n}-m ]+\md =-\md +\md=0\,.
\eeq

\end{proof}

$\square$



\section{Algebroids as Q-manifolds}
\label{Qman}

Any Lie algebra $\g$ gives rise to a complex $(\Lambda^{\cdot} \g^*,
\md_{CE})$, where $\md_{CE}$ denotes the Chevalley-Eilenberg
differential, $\md_{CE}\a(\xi,\xi')=-\a([\xi,\xi'])$ for
$\a \in \g^*$ (and $\xi,\xi' \in \g$)
etc. By definition an element $\a \in \g^*$ is a linear
function on $\g$. Correspondingly, we may identify
$\Lambda^{\cdot} \g^*$ with the space of functions on $\g$ if we
declare the multiplication of two elements $\a, \a'\in \g^*$,
regarded as two functions, to be anticommuting,
$\a \a' = - \a' \a$. This modified law of pointwise multiplication
of functions is denoted by an additional $\Pi$ (indicating parity
reversion), $\Lambda^{\cdot} \g^* \cong C^\infty(\Pi \g)$. In fact,
in what follows it will be important to not only consider a
$\Z_2$-grading, but a $\Z$-grading, inducing the $\Z_2$-grading.
Thus we declare elements of $\g^*$ to have degree 1, and then
$$\Lambda^{\cdot} \g^* \cong C^\infty( \g[1]), $$
the bracket indicating that the canonical degree of the before
mentioned vector space (which in the case of $\g$ is zero)
is shifted by minus
one. Given a basis $e_a$ of elements in $\g$, its
dual basis $\theta^a$ becomes a set of coordinates on  $\g[1]$,
any homogeneous
element $ \omega
\in C^\infty( \g[1])$ can be written as $\omega = \frac{1}{p!}
\omega_{a_1 \ldots a_p} \theta^{a_1} \ldots \theta^{a_p}$, and
the Chevalley-Eilenberg differential
becomes a vector field of degree $+1$ (it
raises the degree of homogeneity of any element by one), which
we will denote by $Q$,
\beq Q = - \frac{1}{2}  \theta^a\theta^bC^c_{ab}
\frac{\partial}{\partial \theta^c} \, , \label{CEQ}
\eeq
where $C^c_{ab}$ denote the structure constants, $[e_a,
e_b] =
C^c_{ab}e_c$. By construction, $Q^2=0$, i.e.~the vector field $Q$
is \emph{homological}.

This construction generalizes to the case where the Lie algebra acts
on a manifold $M$ by vector fields,  $\xi \mapsto \rho(\xi) \in
\Gamma(TM)$. Denoting $\rho(e_a)$ by $\rho_a$, the following
vector field
\beq Q = \theta^a \rho_a - \frac{1}{2}  \theta^a\theta^bC^c_{ab}
\frac{\partial}{\partial \theta^c} \, , \label{actionQ}
\eeq
on $M \times \g[1]$ is again homological. The corresponding cohomology
in degree zero (functions containing no $\theta$s) is obviously
isomorphic to the space of functions on $M$ invariant under the flow
generated by the Lie algebra $\g$. (\ref{actionQ}) and its cohomology may be viewed as
the ``BRST description'' of the space of ``gauge invariant'' functions on
$M$.

In the above $\rho_a$ was a vector field on $M$, so in some local
coordinates $x^i$ of $M$ one has $\rho_a = \rho^i_a \partial_i$, where
$\rho^i_a$ are (local) functions on $M$ (and $\partial_i \equiv
\partial/\partial x^i$). We may now consider to
drop the restriction that $C^c_{ab}$ is constant in  (\ref{actionQ})
but instead also a function on $M$
and pose the question under what conditions on  $\rho^i_a$ and
$C^c_{ab}$ the corresponding vector field  (\ref{actionQ}) squares to
zero. In this context $Q$ is a vector field on a \emph{graded
manifold} ${\cal{M}}$ \footnote{In this article graded manifolds will always signify $\Z$-graded
manifolds or even their special case of $\N_0$-graded manifolds (cf.
also below in the text). The $\Z$-grading induces naturally a
$\Z_2$-grading, which governs the signs of the algebra of
``functions'' on the graded manifold. By the forgetful functor a
$\Z$-graded manifold thus also becomes a particular supermanifold.}
 where the
structure sheaf has local generators $x^i$ and $\theta^a$ of degree
zero and one, respectively; in this case we call  ${\cal{M}}$ of
degree one.  Note that since transition functions are
by definition required to be degree preserving, a change of chart in
${\cal{M}}$ requires $\widetilde \theta^a = M^a_b \theta^b$, with
$M^a_b$ a local function on the \emph{body} $M$ of ${\cal{M}}$;
correspondingly, a graded manifold of degree one is isomorphic to a vector
bundle $E \to M$, $\theta_a$ corresponding to a frame of local
sections. Indicating that fiber--linear coordinates on $E$ have degree
one in the superlanguage, one writes ${\cal M} \cong E[1]$ in this case.

\noindent A general graded manifold is equipped with an Euler
vector field $\e$, such that the grading of the space of functions
$\cF (\cM)$ corresponds to the eigenvalue-decomposition of
$\e$: \beq \cF^k (\cM)&=& \{f\in \cF (\cM) \;|\; \e f=kf\}\;. \eeq
One may consider $k \in \mathbb{Z}$, but we will restrict ourselves to
non-negative integers $k$ except if explicitly stated otherwise.

\noindent Apparently, in the degree one case the extension  of
the algebra of functions on the body, $C^\infty (M) \cong \cF^0(\cM)$, to
the algebra of
functions $\cF(\cM)$ is generated by $\cF^1 (\cM)$, which can be thought of
as the space of sections of a vector bundle (since $\cF^1 (\cM)$
has to be a locally free module over $\cF^0 (\cM)$). If
such a graded manifold is equipped with a degree plus one
homological vector field $Q$, the most general ansatz of which has
the form of (\ref{actionQ}) with now $C^a_{bc}$ being permitted to be functions
on $M$, i.e.~if one considers what is called
a \emph{Q-manifold of degree one} \cite{Schwarz_semiclassical}, the vector
bundle $E \to M$ becomes equipped with the structure of a Lie
algebroid: This may be regarded as a reformulation of
Prop.~\ref{prop:Lie} by the simple identification $C^\infty({\cal
M}) \cong \OE$. Alternatively, one may check directly that $Q^2=0$
implies in homogeneity degree two the morphism property of $\rho$,
Lemma \ref{lem:morphism}, and in degree three the Jacobi condition
(\ref{Loday}), both expressed for the local frame $e_a$. Note
that in this picture the Leibniz rule (\ref{Leibniz}) follows only
from a change of coordinates on $\cal M$, requiring $Q^2=0$ to be
valid in all possible frames; on the other hand, the bracket
becomes automatically antisymmetric when defined by means of
$[e_a,e_b] := C^c_{ab} e_c$, again to hold in any
frame (which turns out to be consistent with the Leibniz rule
(\ref{Leibniz})).\footnote{For a detailed discussion of such type of arguments, which, at least in a slightly more general context, turn out to be more tricky than one may expect at first sight, cf.~\cite{Hansen-Strobl, GruetzmannStrobl}.}

If in a Lie algebroid there exists a frame such that the homological vector field takes the form 
with \emph{constant} structure functions $C^c_{ab}$, then this algebroid is called an \emph{action Lie algebroid}. The bundle $E$ is then isomorphic to $M \times \g$ for some Lie algebra $\g$ acting on $M$. Certainly, in general a Lie algebroid is not of this form and there does not exist a frame, not even locally, such that the structure functions would become constants.

On the other hand one may consider a graded manifold $\CM$
that carries a \emph{homogeneous} symplectic form $\o$ of degree $n$, i.e. \beq L_\e \omega = n \omega \, , \label{Pn} \eeq in which case $\CM$ is called a \emph{P-manifold} of the degree $n$.  Note that the non-degeneracy of the symplectic form then requires that also $\CM$ has degree $n$ (as mentioned, we do not consider graded manifolds with generators of negative degrees here, in which case this statement would no more be true). In the case of $n=1$, we already found that $\cM$ is canonically isomorphic to $E[1]$, $E$ a vector bundle over $M$. It is now easy to see that the P-structure restricts this further \cite{Roytenberg}, $\CM \cong T^*[1]M$,
$(x^i,\theta^a)\cong (x^i,p_i)$, equipped with the canonical
symplectic form $\omega = \md x^i \wedge \md p_i$. Indeed, since the symplectic form is of degree $n=1$, eq.~(\ref{Pn}) implies $\omega =\md\alpha$ with $\a =\imath_\e\omega$. Suppose,
$(x^i, \theta^a)$ are local coordinates of degree 0 and 1,
respectively. Taking into account that $\omega$ is of degree 1, we
immediately obtain that the expression of the symplectic form cannot contain $\md\theta^a\wedge\md\theta^b$, and since $\e = \theta^a \partial/\partial \theta^a$, $\a$ it has to be of the  form: $\a =\a_i^j (x) \theta_j \md x^i$. $\a$ provides a morphism $TM\to E^*$; nondegeneracy of
$\omega$ requires that this is an isomorphism, and thus $(\cM ,\omega)$
is isomorphic to $T^*[1]M$ together with the canonical symplectic
form. Note that functions on $T^*[1]M$ may be identified with
multivector fields. The odd Poisson bracket induced by $\omega$ is
then easily identified with the Schouten-Nijenhuis bracket.

A \emph{PQ-manifold} of degree $n \in \mathbb{N}_0$ is then
simultaneously a $Q$ and a $P$ manifold of degree $n$, such that
$Q$ preserves the symplectic form $\omega$, in which case it turns out
to be even Hamiltonian (cf.~Lemma 2.2 in \cite{Roytenberg}): $Q= \{
\CQ, \cdot \}$  for some function $\CQ$ of degree $n+1$ (since the Poisson bracket decreases the degree by $n$),
$Q^2=0$ reducing to $\{\CQ , \CQ \}=0$.

In general, we use the following sign conventions: the algebra
of differential forms on a graded manifold $\cM$ is defined as $C^\infty (T[1]\cM)$. Then the
degree of $\md h$ is $| h| +1$, where $| h|$ is the degree of $h\in C^\infty (\cM)$.
The Hamiltonian function of a Hamiltonian vector field is obtained from the relation
$\imath_{X_h}\omega =(-1)^{| h| +1}\md h$. The advantage of such a sign convention is that,
if we produce a Poisson bracket by the formula $\{f,h\}=X_f (h)$, then the Lie algebra morphism
property will hold:
$[X_f, X_h]=X_{\{f,h\}}$. For a symplectic structure $\omega$ of degree $n$ written in local Darboux coordinates $\omega =\sum_\alpha \md p_\a\wedge \md q^\a$ the
local Poisson bracket follow to be $\{p_\a, q^\b\}=(-1)^{n| q^\a|}\de_\a^\b$.

Now it is easy to see
that a PQ-manifold of degree one is in one-to-one correspondence with
a Poisson manifold $M$: $\CM \cong T^*[1]M$ and a degree two function
has the form $\CQ \equiv \Pi = \frac{1}{2}\Pi^{ij} p_i p_j$,
corresponding to a bivector field on $M$. The condition that the
Schouten-Nijenhuis bracket of $\Pi$ with itself vanishes, $\{\CQ , \CQ
\}=0$,  is just one way of
expressing the Jacobi condition for the Poisson bracket of functions
(cf.~also examples \ref{Poisson1} and \ref{ex1}). The usual Poisson
bracket $\{ \cdot , \cdot \}_M$ between functions on $M$ is reobtained
here as a derived bracket: $\{ f , g \}_M = \{ \{ f , -\Pi \} , g \}$
for any $f,g \in C^\infty(M) \subset   C^\infty(\CM)$, since the right
hand side indeed yields $\Pi^{ij} f,_i g,_j$ in local coordinates.

We quote a likewise result for the degree two case without proof,
which is due to Roytenberg \cite{Roytenberg}:

\begin{theorem} \label{Roytenberg} A P-manifold of degree two is in  one-to-one
correspondence with a pseudo-Euclidean vector bundle $(E,
\Eg)$.  A PQ-manifold of degree two is in  one-to-one
correspondence with a Courant algebroid $(E,
\Eg,\rho,[\cdot,\cdot])$.
\end{theorem}

We add some remarks on this theorem for illustration. In appropriate Darboux-like
coordinates on
$\CM$, the symplectic form reads \beq
\omega = \md x^i \wedge \md p_i +
\frac{1}{2}\kappa_{ab} \md \theta^a \wedge \md \theta^b \label{55}
\eeq
where $(x^i, \theta^a,p_i)$ are coordinates of degree zero, one, and
two, respectively, the vector bundle $E$ corresponds to the graded
submanifold spanned by $(x^i, \theta^a)$, and
the constants
$\kappa_{ab}$ correspond to the fiber metric $\Eg$ evaluated in some
orthonormal frame $\theta_a$.

The symplectic form is of degree two, so $\CQ$ is necessarily of
degree three and thus of the form
\beq \CQ =  \rho^i_a  \theta^a p_i - \frac{1}{6} C_{abc} \theta^a
\theta^b \theta^c \, , \label{ThetaCour}
\eeq
with coefficient functions $\rho^i_a$ and $C_{abc}$ depending on
$x^i$ only. Sections $\psi$ of $E$ may be identified with functions of
degree one on $\CM$, $\psi = \psi^a \theta_a$, $\theta_a \equiv \kappa_{ab}
\theta^b$, functions $f$ on $M$ with functions of degree zero on
$\CM$. The anchor $\rho$ and the Courant bracket $[ \cdot , \cdot ]$
now follow as derived brackets,
\begin{eqnarray}
\rho (\psi) f  &=& \{ \{ \psi , \CQ \} , f \} \label{anchor} \\
 {[}\psi,\psi'{]} &=&  \{ \{ \psi , \CQ \} , \psi' \} \, ,\label{bracket}
\end{eqnarray}
while the fiber metric comes from the normal Poisson bracket:
$\Eg(\psi,\psi') = \{ \psi, \psi' \}$. From (\ref{anchor}) one obtains
in particular $\rho(\theta_a) x^i = \rho^i_a$ and $[ \theta_a ,
\theta_b ] = C_{ab}{}^c \theta_c$, where the last index in $C_{abc}$
has been raised by means of the fiber metric and $\theta_a$ is
regarded as a local (orthonormal) frame in $E$.

Note that a derived bracket is in general not antisymmetric;
from (\ref{bracket}) one concludes
$[\psi,\psi] = \frac{1}{2} \{ \{ \psi , \psi \} , \CQ \} =
\frac{1}{2} \theta^a \rho_a^i \partial_i \Eg(\psi,\psi)$, which
reproduces the last axiom in
 definition \ref{Cour}. On the other
hand, evaluated in an orthonormal frame $\theta_a$ the bracket does
become antisymmetric, $[ \theta_a ,
\theta_b ] = C_{ab}{}^c \theta_c$.\footnote{The difference to the situation with
$C_{ab}^c$ in the vector field $Q$ of a Lie algebroid is that
there such an equation holds in all frames and also that here a change in
frame results in a different transformation property of $C_{abc}$
in (\ref{ThetaCour}) by lifting this transformation to a canonical
one on $\CM$ (which prescribes a particular induced transformation
property for $p_i$). Cf.~also \cite{Hansen-Strobl} for many more details on this issue.}
 It is also
obvious from (\ref{anchor}) and (\ref{bracket}) that the Courant
bracket satisfies the Leibniz property
 in the second entry of the bracket (while it does not in the first one
 due to the non-antisymmetry).

The study of higher degree PQ manifolds is certainly more involved. They, however, always give rise to Loday algebroids in the following way

\begin{prop} Given a PQ-manifold $\CM$ of degree $n>1$ the functions $\psi, \psi',\ldots$ of degree $n-1$ can be identified with sections in a vector bundle $E$. The formulas (\ref{anchor}) and (\ref{bracket}), where $\CQ$ denotes the Hamiltonian for the Q-structure, equip $E$ with the structure of a Loday algebroid.
\end{prop}
\begin{proof}
First we note that that the functions of any fixed degree $d$ on an $\N_0$--graded manifold are a \emph{locally free module} over the functions of degree zero, which in turn are isomorphic to $C^\infty(M)$. This implies the existence of vector bundles for any of those degrees $d$ over $M$. For $d=n-1$ we call this bundle $E$.

As remarked above, the vector field $Q$ compatible with the symplectic form is always Hamiltonian (even for $n \in \N$); since the (graded) canonical Poisson bracket has degree $-n$, the respective Hamiltonian $\CQ$ has to have degree $n+1$. By an elementary computation adding up respective degrees, one then finds that $\rho(\psi)$ indeed maps functions $f \in C^\infty(M)$ to functions and the bracket $[ \cdot , \cdot ]$ takes again values in the sections of
$E$.\footnote{We denote the expressions on the graded manifold and those isomorphic to them on $E \to M$ by the same symbols.} It is also obvious that $\rho(\psi)$ is a vector field on $M$, since the r.h.s.~of the defining expression (\ref{anchor}) satisfies an (ungraded) Leibniz rule for $f$ being a product of two functions. To have an algebroid (cf.~our definition \ref{algebroid}), we need to verify two things: first, the Leibniz property (\ref{Leibniz}), which follows at once from the two defining expressions on $E$ above and the graded Leibniz property of the Poisson bracket $\{ \cdot, \cdot \}$. Second, that $\rho(\psi)$ is indeed $C^\infty(M)$--linear in $\psi$ so as to really give rise to a bundle map $\rho \colon E \to TM$. The only potentially dangerous term which may violate this condition may arise when $\theta$ has terms quadratic or higher in the momenta $p_i$ conjugate to the coordinates $x^i$ on $M$. However, since the coordinates $p_i$ necessarily have degree $n$ and $\CQ$ has degree $n+1$, this is not possible for $n>1$.

We are left with verifying the Loday property (\ref{Loday}). It is only here where the condition $\{ \CQ, \CQ \}=0$, resulting from $Q^2=0$, comes into the game. We leave the respective calculation, which also makes use of the graded Jacobi identity of the Poisson bracket, as an exercise to the reader. $\square$
\end{proof}

Some remarks: First of all, it is clear from the above proof that any symplectic graded manifold of degree $n\ge 2$ equipped with an arbitrary function $\CQ$ of degree $n+1$ gives, by the above construction, rise to an algebroid structure (in our sense, cf.~Definition \ref{algebroid}). If in addition $\{ \CQ, \CQ \}=0$, this algebroid becomes a Loday algebroid.

Certainly, the higher in degrees we go, the more additional structures arise. Already for the case of $n=2$ the Loday algebroid had the additional structures making it into a Courant algebroid. For higher $n$, however, there will often be other algebroids out of which the Loday algebroid will be composed. Let us illustrate this for $n=3$: In Darboux coordinates the symplectic structure will have the form
\beq \omega = \md x^i \wedge \md p_i +
\md \theta^a \wedge \md \xi_a \, , \label{59}
\eeq
where the coordinates $x^i$, $\theta^a$, $\xi_a$, and $p_i$ have the degrees 0, 1, 2, and 3, respectively. From this we learn that a PQ-manifold of degree 3 has to be symplectomorphic to $T^*[3](V[1])$ with its canonical symplectic form. Here $V \to M$ is a vector bundle and the brackets indicate shifts of degree in the respective fiber coordiantes (note that without a shift the momenta $\xi_a$ conjugate to the degree 1 fiber linear coordinates $\theta^a$ on $V[1]$ have to have degree - 1 while they are now shifted so as to have degree 2). With $\psi = \varphi^a \xi_a + \frac{1}{2} \alpha_{ab} \theta^a \theta^a$, we see that the above vector bundle $E$ in this case is isomorphic to $V \oplus \Lambda^2 V^*$. Now let us consider the Hamiltonian for the Q-structure, i.e.~a function of degree four; it thus have to have the form
\beq \CQ =  \rho^i_a\theta^a - \frac{1}{2}  C^c_{ab}\theta^a\theta^b
\xi_c + \frac{1}{2} \beta^{ab} \xi_a \xi_b + \frac{1}{24} \gamma_{abcd}  \theta^a\theta^b\theta^c\theta^d \, . \label{Q4}
\eeq
Obviously, for $\beta=\gamma=0$ this is nothing but the canonical lift of the Q-structure (\ref{actionQ}) corresponding to a Lie algebroid, thus equipping $V$ with the structure of a Lie algebroid. In the general case $V$ still is an almost Lie algebroid and $E$ can be considered as an appropriate extension into a Loday algebroid. ($V$ itself is \emph{not} a Loday algebroid, having a Jacobiator controlled by a contraction of $\beta$ and $\gamma$. Adding a $V$-2-form $\alpha$ to the section $\varphi$, on which the anchor $\rho$ acts trivially, one can restore the Loday property (\ref{Loday}).)

We know that for $n=1$ we also obtained an algebroid, even a Lie algebroid. However, this was defined on $T^*M$. The above construction leads to a trivial $\R$-bundle over $M$ instead, the sections of which can be identified with functions on $M$, and $\rho \colon C^\infty(M) \to \Gamma(TM)$ is $\R$--linear but no more $C^\infty(M)$--linear, in fact it corresponds (up to a sign being subject of conventions) to the map from functions to their Hamiltonian vector fields. In fact, both equations (\ref{bracket}) and (\ref{anchor}) become equivalent in this degenerate case, equipping $C^\infty(M)$ with the Poisson bracket $\{ \cdot , \cdot \}_M$, as already remarked above. (This bracket defines a Lie algebra structure on the sections of $M \times \R$, but not a Lie algebroid or even a general algebroid structure on this bundle since one does not have an anchor map for it).

To obtain the Lie algebroid structure on $T^*M$ from the (P)Q-manifold $T^*[1]M$ we need to proceed differently. In fact, this provides a procedure that is applicable for any Lie algebroid $E \to M$, corresponding to a degree one Q-manifold $E[1]$. Sections of $E$ are in 1-1 correspondence with vector fields $\psi, \psi',\ldots$ of degree minus one on $E[1]$, and their derived bracket $[[\psi,Q],\psi']$, where the brackets denote the (super)commutator of vector fields, is rather easily verified to reproduce the Lie algebroid bracket between the respective sections.

This procedure can in fact be considered for a  Q-manifold of any degree $n$.  While for general $n$ this leads to what one may call Vinogradov algebroids (Loday algebroids with additional structures, cf.~\cite{GruetzmannStrobl}), for  $n=2$ one obtains a $V$-twisted Courant algebroid as defined in Sec.~2 above.
The degree -1 vector fields define the sections of the bundle $E$, the degree -2 vector fields the sections of the other bundle $V$; the derived bracket $[[\psi,Q],\psi']$ of the degree -1 vector fields defines the Loday bracket on $E$, the ordinary commutator bracket $[\psi,\psi']$, which apparently takes values in the degree -2 vector fields, gives the $V$-valued inner product on $E$. All the defining properties of a $V$-twisted Courant algebroid are rather easy to verify in this case.

On the other hand, a Q-manifold $\CM$ of degree 2 gives always rise to
$V[2] \to \cM\to W[1] $, where the first map is an embedding, setting the degree one coordinates in $\cM$ to zero and the second map a projection, forgetting about the degree two coordinates. (Here $V$ is the same bundle as the one in the previous paragraph and $W$ another vector bundle over the same base manifold). A somewhat lengthy analysis (cf.~\cite{GruetzmannStrobl}) shows that after one has chosen an embedding of $W[1]$ into $\cM$, which composed with the projection giving the identity map, the degree 2 Q-manifold is in bijection with a Lie 2-algebroid. The vector bundle $E$ of the $V$-twisted Courant algebroid picture is then composed of $W$ and $V$, similarly to the situation of the degree 3 PQ manifold discussed above; one finds easily that $E \cong W \oplus W^* \otimes V$, and, under only a few more assumptions (like that the rank of $V$ is at least two) on a $V$-twisted Courant algebroid of this form, also vice versa, the latter arises always from a Lie 2-algebroid in such a way.

We finally remark that if $\CM$ is a Q-manifold of degree $n$, then $T^*[n]\CM$ is a PQ-manifold of the same degree.\footnote{T.S.~is grateful to D.~Roytenberg for this remark in the context of a talk on $V$-twisted Courant algebroids.} E.g.~$T^*[1]E[1]$, with $E[1]$ a Lie algebroid (and $Q$ lifted canonically, certainly), is a degree 1 PQ-manifold, and thus isomorphic to the Lie algebroid of a Poisson manifold. Indeed, 
the dual bundle of a Lie algebroid is canonically a Poisson
manifold,\footnote{The dual $E^*$ of any Lie algebroid $E$ becomes a Poisson
manifold in the following manner: In order to define a Poisson bracket
on $E^*$, it is obviously sufficient to do so on the fiber constant
and fiber linear functions. The former are functions that arise as
pullbacks from functions on the base manifold $M$, the latter are
sections of $E$. It is then straightforward to verify that
\beq \{ f,f' \} = 0 \, , \quad \{ \psi ,f \} =  \rho(\psi) f \, , \quad
\{  \psi , \psi' \} = [ \psi , \psi' ] \, , \label{E*}
\eeq
valid for all $f,f' \in C^\infty(M)$, $\psi,\psi' \in \Gamma(E)$,
defines a Poisson structure on $E^*$.}
in accordance with the easy-to-verify isomorphism $T^*[1]E[1]\cong T^*[1]E$. So, Poisson geometry on $M$ can be viewed as a particular case of Lie algebroid geometry (considering $T^*M$), but also vice versa, a Lie algebroid structure on $E$ as a particular (fiber-linear) Poisson structure (on $E^*$)---and likewise so in higher degrees: A Courant algebroid is a particular case of a V-twisted Courant algebroid and also of a Lie 2-algebroid, but, in an appropriate sense, also Lie 2-algebroids and their corresponding $V$-twisted Courant algebroids, can be viewed as particular Courant algebroids. However, as already the example of a Lie algebroid shows, this is not always the most convenient way of viewing them.

We conclude this section by returning to the question of morphisms of Lie algebroids and, more generally, of any algebroid described by a Q- or a PQ-manifold. A morphism of Q-manifolds is a degree preserving map $\varphi \colon \CM_1 \to \CM_2$ such that its pullback $\varphi^* \colon C^\infty(\CM_2) \to C^\infty(\CM_1)$ is a chain map, i.e.~one has $Q_1 \circ \varphi^* = \varphi^* \circ Q_2$. It is a morphism of PQ-manifolds, if in addition it preserves the symplectic form, $\varphi^* \omega_2 = \omega_1$.

\section{Sigma models in the AKSZ-scheme}
\label{AKSZ}

In this section we want to discuss a particular class of topological sigma models that can be constructed in the context of algebroids. By topological we want to understand that the space of solutions to the classical field equations (the Euler Lagrange equations of the functional) modulo gauge transformations does not depend on structures defined on the base manifold $\Sigma$ in addition to its topology and that for ``reasonable'' topology (the fundamental group of $\Sigma$ having finite rank etc) it is finite dimensional.

In the context of ordinary gauge theories, one such a wellknown space is the moduli space of flat connections on $\Sigma$. A functional producing such a moduli space is the Chern Simons theory
\beq S_{CS}[A]  = \frac{1}{2}\int_\Sigma \kappa\left(A \! \stackrel{\wedge}{,} \!  (\md A+ \frac{1}{3} [A \! \stackrel{\wedge}{,} \! A])\right)   \, , \label{CS}
\eeq
 defined on the space of connections of a trivialized $G$-bundle over an orientable  three-dimensional base manifold $\Sigma$ when specifying $\kappa$ such that $(\g= \mathrm{Lie}(G) , [ \cdot , \cdot ] , \kappa)$ gives a  quadratic Lie algebra; such a connection is represented by a $\g$-valued 1-form $A$ on $\Sigma$, $A \in \Omega^1(\Sigma, \g)$.

To find an appropriate generalization of this theory to the present context, let us first reinterpret the fields $A$ and the field equations $F=0$ of this model within the present context such that it permits a straightforward generalization. First of all, a $\g$-valued 1-form $A$ on $\Sigma$ is evidently equivalent to a degree preserving map (a morphism) \beq a \colon T[1]\Sigma \to \g[1] \, . \label{aYM} \eeq If $e_a$ denotes a basis of $\g$ and $\theta^a$ the linear odd coordinates on $\g[1]$ corresponding to a dual basis, then the 1-forms $A^a \equiv A^a_\mu \vartheta^\mu$ in $A=A^a \otimes e_a$ are given by the pullback of $\theta^a$ with respect to the map $a$, $A^a = a^*(\theta^a)$; here $\vartheta^\mu = \md \sigma^\mu$ are the degree one coordinates on $T[1]\Sigma$ induced by (local) coordinates $\sigma^\mu$ on $\Sigma$. This is easy to generalize, given the background of the previous sections: As our generalized gauge fields we will consider morphisms of graded manifolds
\beq a \colon \CM_1 \to \CM_2 \label{a} \eeq
keeping $\CM_1 = T[1] \Sigma$ (so these gauge fields will be a collection of differential forms on $\Sigma$ of various form degrees) and taking $\CM_2$ as a more general $\N_0$-graded manifold than $\g[1]$.

In the example, not only $\CM_1 = T[1] \Sigma$ but also $\CM_2 = \g[1]$ is not any graded manifold, but even a Q-manifold, $Q_1$ is the de Rham differential and $Q_2$ the Chevalley-Eilenberg differential (\ref{CEQ}). Let us reconsider the field equations $F=0$ with $F=F^a \otimes e_a$ the curvature (\ref{F}) of $A$ in this context.  Obviously $F^a=0$, iff $\md A^a=-C^a_{bc} A^b \wedge A^c$. With the identification of the vector fields $Q_1$ and $Q_2$ above, this in turn is now seen to be $Q_1 a^* (\theta^a) = a^* Q_2 (\theta^a)$. In other words, the map $a^* \colon C^\infty(\CM_2) \to C^\infty(\CM_1)$ needs to be a chain map, or, according to the definition of a Q-morphism, the field equations of the Chern Simons gauge theory express that (\ref{a}) is not only a morphism of graded manifolds but even a Q-morphism (a morphism of differential graded manifolds).

We are thus searching a functional defined on morphisms (\ref{a}) such that its Euler-Lagrange equations forces these to become morphisms
\beq a \colon (\CM_1,Q_1) \to (\CM_2,Q_2) \label{aQ} \, . \label{fields} \eeq
Certainly, for a true generalization of the \emph{gauge} theory defined by means of (\ref{CS}), we also need to reinterpret its gauge transformations appropriately, so that we can formulate also the desiderata for the gauge symmetries of the searched for action functional. It turns out that on the solutions of (\ref{aQ}) the gauge symmetries receive the interpretation of Q-homotopy.\footnote{We refer to \cite{BKS} for the details.} We are thus searching for a functional defined on (\ref{a}) such that the moduli space of classical solutions modulo gauge transformations is the space of Q-morphisms from $(T[1]\Sigma,\md)$ to the target Q-manifold $(\CM_2,Q_2)$ modulo Q-homotopy.

A functional can be obtained by the so-called AKSZ-method \cite{AKSZ} (cf.~also \cite{CFAKSZ} and \cite{RoytenbergAKSZ}) for the case that the target carries also a compatible symplectic form, i.e.~that the target is a PQ-manifold. However,
this method yields in fact already the BV-extension of the searched-for (``classical'') functional. We thus briefly recall some basic ingredients of the BV-formalism and we will do this at the example of a toy model, where we consider functions instead of functionals, as well as for the Chern-Simons theory (\ref{CS}) above.

The toy model is the following one: Consider a function $S_{cl}$ on a manifold $M$ invariant with respect to the action of some Lie algebra $\g$: $v(S_{cl})(x) = 0$ for all points $x \in M$ and all elements $v \in \Gamma(TM)$ corresponding to the action of an element of $\g$. We learnt that $M \times \g$ carries the structure of a Lie algebroid, the action Lie algebroid. Let us thus, more generally, consider a Lie algebroid $E \to M$ together with a function $S_{cl} \in C^\infty(M)$ constant along the Lie algebroid orbits on $M$, i.e.~s.t.~$\rho(\psi) S_{cl} = 0$ for all $\psi \in \Gamma(E)$.

$S_{cl}(x)$ is supposed to mimic a functional invariant w.r.t.~some gauge transformations (in the more general sense, cf., e.g., \cite{HenneauxTeitelboim}, also for general details on the BV formalism). Consider for simplicity first the case of the action Lie algebroid again. The gauge invariance here corresponds to $\delta_\epsilon S_{cl}=0$ where $\delta_\epsilon x^i = \epsilon^a \rho_a^i(x)$ with some arbitrary $\g$-valued parameters $\epsilon = \epsilon^a e_a$ (and $\rho(e_a)=\rho_a^i(x)\frac{\pr}{\pr x^i}$). In the BRST--BV approach one first replaces the parameters $\epsilon^a$ by anticommuting variables $\theta^a$ so that the action of the BRST operator $\delta$ on the original, classical fields (coordinates here) becomes $\delta x^i = \theta^a \rho_ a^i(x)$. This is completed by the action of $\delta$ on the ``odd parameters'' $\theta^a$, $\delta \theta^a = - \frac{1}{2} C^a_{bc} \theta^b \theta^c$ rendering $\delta$ nilpotent, $\delta^2 = 0$. In fact, this BRST operator is evidently
nothing but the operator $Q$ defined on $E[1]$, $\delta \equiv Q$, cf.~eq.~(\ref{actionQ}), and this works for a general Lie algebroid $E$ with
\beq\rho(e_a)=\rho_a^i
(x)\frac{\pr}{\pr x^i}\;, \hspace{3mm}  [e_a, e_b]=C_{ab}^c (x) e_c\;,\eeq
where now $e_a$ is a local frame of sections of $E$ and $C_{ab}^c$ thus became structure functions (instead of just structure constants).

To obtain the BV-form of the action, one now turns to the (graded) phase space
version of this, i.e.~one introduces momenta (shifted in degree, cf.~below) for each of the fields (coordinates on $E[1]$ in our example), which, conventionally, are called the \emph{antifields}. In our example $x_i^*$ conjugate to $x^i$ and $\theta_a^*$ conjugate to $\theta^a$. Adding the Hamiltonian lift of $\delta$ to the classical action we obtain
\beq S_{BV}&=& S_{cl}(x) +
\t^a\rho_a^i (x) x^*_i-\frac{1}{2}C_{ab}^c (x) \t^a \t^b
\t^*_c\;.
\label{SBVtoy}
\eeq
To have this to have a uniform total degree, we need to shift the momenta in degree by minus one, $\deg x^*_i=-1$, $\deg \t^*_a=-2$. Thus the BV odd phase space we are looking at together with the BV-function are of the form
\beq \CM_{BV} = T^*[-1]E[1] \, ,  \quad \omega_{BV} = \md x^*_i \wedge \md x^i + \md \t_a^* \wedge \md \t^a \, , \qquad S_{BV} = S_{cl} + \CQ \, ,\label{BVtoy} \eeq
where $\CQ$ is the Hamiltonian of the canonical Hamiltonian lift  of the vector field $Q$ of the Lie algebroid $E[1]$, the odd Poisson bracket, the socalled BV-bracket $\{ \cdot , \cdot \}_{BV}$ has degree +1, so that $Q_{BV}= \{ S_{BV},\cdot \}_{BV}$ has degree +1 as well. 
$Q_{BV}$ is a differential, i.e.~$S_{BV}$ satisfies the so-called \emph{classical master equation}
\beq \{ S_{BV} , S_{BV} \}_{BV} = 0 \, ,\label{master} \eeq
which is completely obvious from our perspective: $\CQ$ Poisson commutes with itself since it is the Hamiltonian for the Lie algebroid differential $Q$ (there are no ``odd constants'', $\{ \CQ , \CQ \}_{BV}$ having degree 1), $S_{cl}(x)$ Poisson commutes with itself since it depends on coordinates only, and $\{S_{cl}(x),\CQ\}=0$ since it corresponds to the Lie algebroid action on $S_{cl}$, which is zero by assumption.

So, $(\CM_{BV}, \omega_{BV}, Q_{BV} = ( S_{BV},\cdot )_{BV})$ defines a PQ-manifold. In contrast to the previous $\Z$-graded Q-manifolds, this PQ-manifold also has negative degree generators. In fact, it is a cotangent bundle of an $\N_0$-graded Q-manifold with a shift in the cotangent coordinate degrees such that it is precisely the momenta (antifields) that have negative degrees. The total degree is called the \emph{ghost number} conventionally; so the classical fields (coordinates $x$ in the toy model) have ghost number zero, the ``odd gauge parameters'' $\t^a$, the ghosts, have ghost number one, and the antifields have negative ghost number. We can also just consider the number of momenta or antifields: denoting this number by a subscript, we see that  $S_{BV} = S_0 + S_1$ here, where $S_0 = S_{cl}$ and $S_1=\CQ$.

In general, the BV formalism is more involved, there can be terms of higher subscript. Still, always $S_0$ is the classical action $S_{cl}$. Also we see that $\{x^*_i, S_{BV}\}_{BV}|_0 = \frac{\partial S_{cl}}{\partial x^i}$.
Applying $Q_{BV}$ to the classical fields and setting the antifield-less part to zero yields the critical points of $S_{cl}$, i.e.~the classical field equations.

We now turn to the BV formulation of the Chern-Simons gauge theory (\ref{CS}).
As before in (\ref{SBVtoy}), we add to the classical action (\ref{CS}) a term linear in the classical antifields $A_a^*$ with a coefficient that is the (infinitesimal) gauge transformations, replacing the gauge parameters by odd fields $\beta^a$ (so, naturally, $A^*_a$ should be a 2-form on $\Sigma$), and we complete the expression by a term proportional to the odd (anti)fields $\beta_a^*$, 3-form on $\Sigma$, containing the structure constants such that the master equation (\ref{master}) is satisfied;\footnote{Again, also this example is a relatively simple one for what concerns the BV-formalism and the simpler BRST approach would be sufficient to yield the same results. However, already the models generalizing the Chern-Simons theory that we will discuss below, like the Poisson sigma model or the more general AKSZ sigma model, have a more intricate ghost and antifield structure.}  this yields
\beq
S_{CS-BV}[A,A^*, \beta,\beta^*]= S_{CS}[A] + \int_\Sigma (\md \beta^a + C^a_{bc} A^b \beta^c) A^*_a + \frac{1}{2} C^c_{ab} \b^a \b^b \b_c^*
\label{SCSBV}\eeq
which is in a striking similarity with our toy model (\ref{SBVtoy}). Remembering that we could rewrite the toy model in a much more elegant form using the Q-language, cf.~(\ref{BVtoy}), we strive for a similar simplification in the present context.

For this purpose we first recall that the quadratic Lie algebra used to define the Chern-Simons theory is a Courant algebroid over a point, which in turn is a degree 2 PQ-manifold over a point (cf.~Theorem \ref{Roytenberg}):
\beq \CM_2 = \g[1]\, , \quad \omega = \frac{1}{2}\kappa_{ab} \md \theta^a  \wedge \md \theta^b \, , \qquad  \CQ =  -\frac{1}{6} C_{abc} \theta^a
\theta^b \theta^c \, . \label{CStarget}
\eeq
This is the target of the map (\ref{a}), with the source, $\CM_1=T[1]\S$, being a Q-manifold ($Q_1=\md$). The map (\ref{a}) corresponds to the classical fields $A$, which we amended with further fields $\b, \b^*, A^*$  above.

It is tempting to collect all these fields together into a super field (indices were raised by means of $\kappa$)
\beq \CA^a = \beta^a + A^a + A_*^a + \b_*^a \label{CA}
\eeq
by adding them up with increasing form degrees. In fact, this corresponds to an extension of the morphism (\ref{a}) to what is called a \emph{map} $\tilde a$ from $\CM_1$ to $\CM_2$. But before commenting on this extension on this more abstract level, we want to first convince ourselves that the concrete expression (\ref{CA}) is useful. Let us consider the action (\ref{CS}) simply replacing $A$ by $\CA$--and clearly keeping only the top degree forms for the integration over $\S$. Viewing differential forms on $\S$ as graded functions on $T[1]\S$, we can also write this as a Berezin integral over that graded manifold and we will partially do so in what follows. Let us consider the first part of this action first:
\beq S_{source}[\CA] = \frac{1}{2} \int_{T[1]\S} \kappa(\CA \! \stackrel{\wedge}{,} \!  \md \CA) =  \frac{1}{2} \int_\Sigma \kappa(\b ,  \md A^*) + \kappa(A \! \stackrel{\wedge}{,} \!  \md A) +  \kappa(A^* \! \stackrel{\wedge}{,} \!  \md \b)  \, . \label{Ssource}
\eeq
We see that taking $\S$ to have no boundary for simplicity, with appropriate sign rules (see \ref{AKSZSM}, in general)
, the first and the third term become identical and they reproduce all the terms containing a $\md$ in the BV-action above. Similarly,
\beq S_{target} [\CA] = \frac{1}{6} \int_{T[1]\S} \kappa ( \CA \! \stackrel{\wedge}{,} \! [\CA \! \stackrel{\wedge}{,} \! \CA])) \equiv \frac{1}{6} \int_{T[1]\S} C_{abc} \, \CA^a \wedge \CA^b \wedge \CA^c \label{Starget}
\eeq
is easily seen to reproduce the remaining terms in (\ref{SCSBV}).
So we see that
\beq S_{CS-BV}[A,A^*, \beta,\beta^*]= S_{CS}[\CA] = S_{source}[\CA] + S_{target}[\CA] \label{CSCA}\eeq
and we will strive at understanding this action in analogy to the toy model situation (\ref{BVtoy}).

For this purpose we need to identify the BV phase space of this situation. We thus first return to a further discussion of the superfields (\ref{CA}). The notion of a smooth map of graded
manifolds is an extension of the notion of a morphism:\footnote{Strictly speaking, the first and consequently the third formula do not make sense: $\Maps(\CM_1,\CM_2)$ turns out to be an infinite dimensional graded manifold and, as any graded manifolds, only its degree zero part, the body, contains points; a graded manifold, like a supermanifold, is not even a set. Still, it is useful to \emph{think} like this; like everything in supergeometry, the real definitions are to be given algebraically on the dual level.}
\beq \widetilde a \in \Maps(\CM_1,\CM_2) \; , \quad   a \in \Mor(\CM_1,\CM_2) = \Maps_0(\CM_1,\CM_2) \; , \quad P_0 \widetilde a = a\, .\label{tildea} \eeq
In the case of a flat
target manifold like $\g[1]$ the description of a smooth map is rather clear from the example (\ref{CA}): we (formally) allow the functional dependence of the target coordinates on arbitrary degrees on the source. $\Maps(\CM_1,\CM_2)$ is naturally graded: The coefficients in the expansion like (\ref{CA}) are coordinates on this map space. Since the coordinates on the target $\CM_2=\g[1]$ have degree one, each term in the expansion (\ref{CA}) has degree 1 as well. Correspondingly, the ghost $\b^a$ has degree 1, $A_\mu^a(\sigma)$, the coefficient in $A^a = A_\mu^a \vartheta^\mu$ has degree zero (since $\vartheta^\mu = \md \sigma^\mu$, coordinates on the source $T[1]\Sigma$,  have degree 1), $(A^a_*)_{\mu \nu}(\sigma)$, the coefficients of the 2-form field $A^a_*$, are fields of degree -1 etc. A field (or antifield) is the same as a coordinate on $\Maps(\CM_1,\CM_2)$. Its degree zero part $\Maps_0(\CM_1,\CM_2)$ is the space of morphisms or the space of the original classical maps $a$, which correspond to the (classical) fields $A_\mu^a(\sigma)$. $a$ results from $\widetilde a$ by projection (formal operator $P_0$ in (\ref{tildea})) to its degree zero part, which are those maps that are degree preserving: in the example (\ref{CA}) this is keeping the second term.

Generally, the space of maps between graded manifolds $\cM_1$ and
$\cM_2$, denoted as $\Maps (\cM_1, \cM_2)$, is uniquely determined
by the functorial property (cf. eg. \cite{Deligne-Morgan} or also
\cite{Roytenberg}): for any
graded manifold $Z$ and a morphism $\psi: Z\times\cM_1\to\cM_2$
there exists a morphism $\tilde{\psi}: Z\to \Maps (\cM_1, \cM_2)$
such that $\psi =\widetilde{ev}\circ (\tilde{\psi}\times\mathrm{Id} )$, where the evaluation map $\widetilde{ev}$ is (formally)
defined in the obvious way:
\beq \widetilde{ev} \colon \Maps(\CM_1,\CM_2) \times \CM_1 \to \CM_2 \, , \; (\widetilde a, \widetilde \sigma) \mapsto \widetilde a (\widetilde \sigma)  \, .
\label{widetildeev} \eeq
Though the map $\widetilde a$ underlying (\ref{CA}) is not a morphism of graded manifolds (it is does not even induce an ungraded
morphism of the associative algebras of functions!), the
evaluation map $\widetilde{ev}$ is.

So the BV-phase space is $\CM_{CS-BV}=\Maps(T[1]\Sigma,\g[1])$, or, more generally, $\CM_{BV} = \Maps(\CM_1,\CM_2)$. It is canonically an odd (infinite dimensional, weakly) symplectic manifold. The symplectic form $\o_{BV}$ is induced by means of the one on the target, eq.~(\ref{CStarget}):
\beq \o_{BV} = \int_{T[1]\S} \widetilde{ev}^* \omega = \frac{1}{2} \int_{T[1]\Sigma} \kappa_{ab} \,  \delta \CA^a \wedge \delta \CA^b  \, ,  \label{oBV} \eeq
where $\delta$ denotes the de Rahm differential on $\CM_{BV}$ (so as to clearly distinguish it from the \emph{vector field} $\md$ on $\CM_1$!). Note that the pullback by means of the evaluation map (\ref{widetildeev}) of the 2-form $\omega$ produces a (highly degenerate) 2-form on $\CM_{BV} \times \CM_1$, which, moreover, still has degree 2 since $\widetilde{ev}$ is a morphism, as remarked above. The Berezin integration over $T[1]\Sigma$ then reduces the degree of the resulting differential form by 3 (since $\Sigma$ is three-dimensional), so that $\o_{BV}$ has degree -1, such as in our toy model (\ref{BVtoy}).
We could write out the right hand side of  (\ref{oBV}) similarly to (\ref{Ssource}); this then makes it clear that, after the integration, the resulting 2-form is indeed (weakly) non-degenerate. Let us stress at this point that certainly this construction would \emph{not} work when
sticking to the purely classical fields $A$: A likewise expression $\int_{\S} \kappa(\delta A \! \stackrel{\wedge}{,} \! \delta A)$ would be nonzero only for a \emph{two-}dimensional surface $\S$---in fact, this then is the symplectic form of the classical phase space of the Chern Simons theory. So, also at this point the extension from $a\in \Mor(\CM_1,\CM_2)$ to $\widetilde a\in \Maps(\CM_1,\CM_2)\equiv \CM_{BV}$ is essential, because only the latter space is naturally symplectic, and indeed the symplectic form has degree -1 so that the BV-bracket will have degree +1, as it should be.

From the above the first steps in the generalization of the Chern-Simons BV action to a more general setting is clear: We will keep $(T[1]\Sigma, \md)$ as our source $Q$-manifold $(\CM_1,Q_1)$, with $\Sigma$ having a dimension $d$ different from three in general. As target we need to choose at least a symplectic graded manifold, but in fact, like in our guyding example (\ref{CStarget}), we will consider a PQ-manifold $(\CM_2, \omega_2, Q_2)$, of degree $n$ in general (for $n>0$ we can also replace the symplectic vector field $Q_2$ by its generating Hamiltonian function $\CQ_2$ of degree $n+1$). Now   $\int_{T[1]\S} \widetilde{ev}^* \omega_2$ gives a degree -1 symplectic 2-form on $\CM_{BV} = \Maps(\CM_1,\CM_2)$, iff $d=n+1$ (since the Berezin intergration reduces the degree of the 2-form by $d$).

It remains to rewrite the action (\ref{SCSBV}) or (\ref{CSCA}) in a form that will resemble somewhat the BV-function of the toymodel (\ref{BVtoy}). In particular, according to our assumptions on source $\CM_1=T[1]\Sigma$ and target (\ref{CStarget}), we are having a vector field $Q_1=\md$ and $Q_2=\{ \CQ, \cdot \}$ on the source and the target, respectively. Both of them give rise to a vector field on $\CM_{BV} = \Maps(\CM_1,\CM_2)$:\footnote{Details for the remaining part of the paragraph and the following one can be found in \cite{AKSZ}, \cite{CFAKSZ}, and \cite{RoytenbergAKSZ}.} Identifying  the tangent space at $\widetilde a \in \CM_{BV}$ with $\Gamma(\CM_1, \widetilde a^* T\CM_2)$, the two vector fields give rise to
\beq  \widetilde a_* Q_1  \qquad \mathrm{and}  \qquad Q_2 \circ
\widetilde a \, , \label{Q12} \eeq
respectively. These two vector fields on $\CM_{BV}$, both of degree 1, are (graded) commuting, as acting on the right and the left of the map $\widetilde a$. Their difference
\beq \widetilde f : = \widetilde a_* Q_1  - Q_2 \circ  \widetilde a \label{tildef} \eeq
is nothing but the BV Operator $Q_{BV}$. It turns out that both vector fields
are Hamiltonian with respect to the symplectic form (\ref{oBV}), with the Hamiltonians being given (\ref{Ssource}) and (\ref{Starget}), respectively. In the spirit of (\ref{BVtoy}), we can now also bring the BV action into the form 
\beq
S_{BV-AKSZ}=\int\limits_{T[1]\Sigma}\imath_{\md}\widetilde{ev}^* (\a)+(-1)^{d}\widetilde{ev}^*
\CQ\;, \label{AKSZSM} \eeq
which defines the general AKSZ sigma model. Here $d$ is the dimension of $\Sigma$, in the example of the Chern-Simons theory thus $d=3$ and $\CQ$ is the Hamiltonian (\ref{CStarget}). In the first term $\a$ is a primitive of $\omega$, $\omega = \md \a$ and $\imath_{\md}$ denotes the contraction with $Q_1 = \md$.  With $\Sigma$ having no boundary, the action (\ref{AKSZSM}) is independent of the choice of $\a$.
If $q^\alpha$ denote Darboux coordinates of the target PQ-manifold, which, as mentioned, has degree $d-1$, i.e.~$\omega = \frac{1}{2} \omega_{\a\b} \md q^\a \wedge \md q^\b$ with $\omega_{ab}$ being constants, we can choose $\alpha = \frac{1}{2} \omega_{\a\b} \, q^\a \wedge \md q^\b$. Let $\CA^\alpha = \widetilde{a}^*(q^\a)$. Then we can (somewhat formally) rewrite the AKSZ action ``evaluated'' at $\widetilde a$ more explicitly as
\beq S_{BV-AKSZ}[ \CA] =  \int_{T[1]\Sigma}  \frac{1}{2}\omega_{\a\b} \CA^\a \md \CA^\b +(-1)^d \widetilde a^* \CQ \,  . \label{AKSZBV} \eeq
In this form it is very easy to see that we reproduce from this the Chern-Simons theory in its form (\ref{CSCA}) upon the choice   (\ref{Starget}) together with $d=3$.

In general, the AKSZ sigma model is defined for a degree $d-1$ PQ-manifold on a $d$ dimensional base. The classical action results from the BV form of it simply by replacing $\widetilde a$ by its degree zero part $a$, i.e.~with $A^\a = a^* q^\a$
\beq S_{AKSZ}[A] =  \int_{T[1]\Sigma}  \frac{1}{2}\omega_{\a\b} A^\a \md A^\b +(-1)^d a^* \CQ \, . \label{AKSZclass} \eeq
Again, in the Chern-Simons case we easily find the classical aciton (\ref{CS}) reproduced.

We can, however, now also parametrize this sigma model more explicitly by means of the considerations in section \ref{Qman} for the lowest dimensions of $\Sigma$: For $d=2$, we need to regard degree 1 PQ-manifolds, which we had found to be always of the form
\beq \CM_2 = T^*[1] M \quad , \quad \omega=\md p_i \wedge \md x^i \qquad , \quad \CQ = \frac{1}{2} \Pi^{ij} p_i p_j \eeq
with $\Pi$ a Poisson bivector (cf.~example \ref{ex1}); i.e.~the target data are uniquely determined by a Poisson manifold $(M,\Pi)$. The most general AKSZ sigma model for $d=2$ is thus seen to be the Poisson sigma model \cite{PSM}, \cite{Ikeda}:
\beq S_{PSM}[X^i,A_i] = \int_\Sigma A_i \wedge \md X^i + \frac{1}{2} \Pi^{ij}(X) A_i \wedge A_j \label{PSM} \eeq
where $X^i = a^* x^i$ are 0-forms on $\Sigma$, $A_i = a^* (p_i)$ 1-forms, and we used a more standard notation of integration over the (orientable) base manifold $\Sigma$.

For $d=3$ we see that we get a (topological) sigma model for any Courant algebroid, cf.~Theorem \ref{Roytenberg}, and the corresponding sigma model is easily specified by means of (\ref{55}) and (\ref{ThetaCour}). We call it the Courant sigma model. It was obtained first by Ikeda in \cite{CS-Ikeda} and later by Roytenberg more elegantly by the present method \cite{RoytenbergAKSZ}.

For $d=4$ the geometrical setting of the target, a degree 3 PQ manifold, has not yet been worked out in detail or given any name. But again we can write down the explicit sigma model in this case using (\ref{59}) and (\ref{Q4}).

How we presented the AKSZ sigma model, its main purpose is to find a topological action functional such that its classical field equations are precisely
Q-morphisms, (\ref{aQ}). While this was proven for the Poisson sigma model explicitly in \cite{BKS}, the present formalism permits an elegant short and general proof. As the comparison of (\ref{AKSZclass}) with (\ref{AKSZBV}) shows, the difference between the classical action and its BV-extension is that we merely have to perform the replacement (\ref{fields}) to go from one to the other. This is a very specific feature of the present topological models, the BV extension is usually not that simple to obtain for a general gauge theory; here, however, it works like this as we saw above. 
The BV-AKSZ functional (\ref{AKSZBV}) becomes stationary precisely when (\ref{tildef}) vanishes (since this is its Hamiltonian vector field and the symplectic form is non-degenerate).  Correspondingly, the variation of (\ref{AKSZclass}) results into the same equation, but where $\widetilde a$ is replaced by $a$, so the Euler Lagrange equations are equivalent to the vanishing of
\beq  f \colon \CM_1 \to T[1]\CM_2 \quad , \qquad
f : = a_* Q_1  - Q_2 \circ  a \; , \label{f} \eeq
where we shifted the degree of the tangent bundle to the target so as to have $f$ being degree preserving like $a$. This in turn is tantamount to the chain property of $a^*$ (cf.~also Lemma \ref{fstar} below).


We finally remark that also Dirac structures and generalized complex structures can be formulated with profit into the language of super geometry (the former ones as particular Lagrangian Q-submanifolds in the degree 2 QP-manifold $T^* [2]T[1]M$), which in part can be used also to formulate particular sigma models for them within the present scheme not addressed in the present article (and different from those of the following section). We refer the reader for example to 
\cite{Grabowski0601761} and \cite{hepth09110993}.

\section{Sigma models related to Dirac structures}
\label{DSM}

As we have mentioned before, a Poisson manifold gives
a particular example of a Dirac structure, determined by the graph
of the corresponding bivector in $TM\oplus T^*M$.
Similar to the Poisson sigma model (PSM), the target space of which
is a Poisson manifold, we now want to consider a topological sigma model  associated to any Dirac subbundle of an exact Courant algebroid. This Dirac sigma model (DSM) \cite{KSS} is supposed to be at least equivalent to the PSM
for the special choice of a Dirac structure that is the graph of a Poisson bivector. Also, we want to continue pursuing our strategy that its classical field equations
should be appropriate morphisms. In fact, in lack of a good notion of a morphism of a Dirac structure, we will content ourselves with a Q-morphism again, i.e.~a Lie algebroid morphism in this  case (since any Dirac structure is in particular a Lie algebroid structure). The model will be two-dimensional. As we saw in the previous section that  the most general two-dimensional model obtained by the AKSZ scheme is the PSM, the DSM does not result from this method, at least not by its direct application.

The target space of the Dirac sigma model is a manifold
together with a Dirac structure $D$ in an exact Courant algebroid
twisted by a closed 3-form $H$. The space-time is a 2-dimensional
surface $\Sigma$. We need also some auxiliary structures---a Riemannian
metric $g$ on $M$ and a Lorentzian metric $h$ on $\Sigma$.
 A classical field of the DSM is a bundle
map $T\Sigma\to D$. First this corresponds to the base map $X\colon \Sigma\to M$, which is corresponds to a collection of ``scalar fields''. Taking into account that $D\subset TM\oplus T^*M$, we
represent the remaining field content by a couple of sections $V\in \Omega^1 (\Sigma, X^*TM)$, satisfying the constraint that this couple combines into a section of $T^*\Sigma\otimes X^*D$.
Here $\Omega^p (N, E)$ is by definition $\Gamma (\Lambda^p T^*N\otimes E)$ for any smooth
manifold $N$ and vector bundle $E\to N$. The DSM action is the sum of two terms, one using the auxiliary geometrical structures and one that uses the topological structures only:  $S_{DSM}[X,V,A]=S_{geom}+S_{top}$, where
\beq\label{DSM-kin}
S_{geom}&\colon =& \frac{\a}{2}\int\limits_{\Sigma}\mid\mid \md X-V\mid\mid^2\,, \label{Sgeom}\\
\label{DSM-top}
S_{top} &\colon =&\int\limits_{\Sigma} \left(\langle A \stackrel{\wedge}{,} \md X  - \frac{1}{2} V\rangle \right) +\int\limits_{N^3} H\,. \label{Stop}
\eeq
Here $\a$ is a real (non-vanishing, in general) constant\footnote{The Lorentzian
signature of $h$ is chosen for simplicity. In the Riemannian version of the Dirac sigma model  the coupling
constant $\a$ has to be totally complex, that is, $\a\in i\R$.}, the absolute value in the first term corresponds
to the canonical pairing by $h \otimes g$ on $T\Sigma \otimes
TM$,\footnote{More explicitly this expression was defined after
formula (1.1).}, the
brackets $\langle, \rangle$ denote the pairing between $TM$ and
$T^*M$, and the last term in the topological part of the action is the integral
of $H$ over an arbitrary map $N^3\to M$ for $\partial N^3 =\Sigma$ which extends
$X\colon \Sigma\to M$ (here we assumed for simplicity that $X$ is homotopically trivial\footnote{For less a topologically less restrictive setting, one permits action functionals up to integer multiplies of $2 \pi \hbar$---here we refer to the literature on Wess-Zumino terms, cf.~\cite{KSS} and references therein.}, which is e.g.~always the case if the second homotopy group of $M$ is zero).

This sigma model generalizes the G/G Wess-Zumino-Witten (GWZW) \cite{WZW, GG} and
($H-$twisted) Poisson sigma model, simultaneously. We first comment on the
relation to the PSM. Let us first choose $H=0$ (cf.~Example \ref{ex1}) and $\alpha = 0$ (yielding $S_{geom} \equiv 0$); with $D$ being the graph of a bivector $\Pi$, we have $V= X^* \Pi (A, \cdot)$ and comparison with (\ref{PSM}) shows that we indeed have $S_{DSM}=S_{PSM}$ in that case.
For $H$ non-zero we get the twisted version of the Dirac structure of example \ref{ex1} and by means of eq.~(\ref{Stop}) of the corresponding sigma model \cite{Ctirad}, respectively. (We will comment on non-zero $\alpha$ below). The GWZW model results from a special choice of the Dirac structure in an exact Courant algebroid on a quadratic Lie group $M=G$, $H$ is the Cartan 3-form, $g$ the biinvariant Riemannian metric on $G$, and $\alpha=1$. In the description of a Dirac structure by an orthogonal operator $S \in \End(TM)$ once a metric $g$ has been fixed on $M$, cf.~Cor.~\ref{oneone}, the Dirac structure on $G$ is the one given by the adjoint action $S=\Ad_g$, where $g \in M$ is the respective base point.

The GWZW model is in so far an important special case as it is not only a well known model in string theory (cf., e.g., \cite{GG}), but it is also \emph{known} to be topological, despite the appearance of auxiliary structures needed to define it. Also, its Dirac structure can be shown to  \emph{not} be a graph, e.g.~by employing the characteristic classes described in sec.~\ref{sec:Dirac}.

Part of the topological nature for the general DSM can be already verified on the level of the classical field equations: Do they depend on the auxiliary structures like $g$ or $h$?
Here we cite the following result from \cite{KSS}:
\begin{theorem} Let $\a\ne 0$, then a field $(X,V,A)$ is a solution of the equations of motion, if and only if the corresponding bundle map $T\Sigma\to D$ induces a Lie algebroid morphism.
\end{theorem}
Since the notion of a Lie algebroid morphism does not depend on auxiliary structures as those mentioned above, we see that at least this condition is satisfied for nonvanishing $\alpha$.

The proof of this theorem is somewhat lengthy,\footnote{One of the complications is that the fields $A$ and $V$ are not independent from one another in general and this has to be taken care of when performing the variations. One way of doing that is by using the orthogonal operator $S$ mentioned in Cor.~\ref{oneone}, which permits us to express these two fields by means of an independent $W \in \Omega^1(\S,X^*TM)$ according to
$A = (1+S)W$ and $V=(1-S)W$.} so that we do not want to reproduce it here; instead we want to prove it for the simplest possible Dirac structure, $D=TM$, example \ref{TM} above, and refer for the general fact to \cite{KSS}.

We start by calculating the field equations of the sigma model; but for $D=TM$ and $H=0$ (cf.~example \ref{TM}) the topological part of the action is identically zero. So, it remains to look at the variation of (\ref{Sgeom}). Since $V$ is an independent (unconstrained) field in this case, the variation of the quadartic term w.r.t.~$V$ yields
\beq \md X = V \label{dxv} \eeq
while the $X$-variation vanishes on behalf of that equation. Mathematically, this equation is tantamount to saying that the vector bundle morphism $a \colon T \Sigma \to T M$ is the push forward of a map $X \colon \Sigma \to M$, $a = X_*$. We obtain the required statement in this special case by use of the following
\begin{lemma} A Lie algebroid morphism from the standard Lie algebroid over a manifold $\Sigma$ to the standard Lie algebroid over a manifold $M$ is the push forward of a smooth map $X \colon \Sigma \to M$.
\end{lemma}
\begin{proof} Recall from section \ref{Qman} that the Lie algebra morphism above can be defined best by a degree preserving map $\bar a \colon T[1] \S \to T[1]M$ such that $\bar a^* \colon C^\infty(T[1]M) \to  C^\infty(T[1]\S)$ commutes with the respective differentials characterizing the Lie algebroid, here being just the respective de Rahm differentials.  Let us choose local coordinates $\sigma^\mu, \vartheta^\mu = \md  \sigma^\mu$ and $x^i, \theta^i = \md x^i$, respectively. Then $X^i = \bar a^* x^i$ corresponds to the base map $X$ of the Lemma. On the other hand $V^i := \bar a^* \theta^i = \md \bar a^* x^i = \md X^i$, where in the second equality we used that $\bar a^*$ commutes with $\md$. $V^i = \md X^i$ or, equivalently, $V^i_\mu = X^i,_\mu$ is the searched-for equation.
\end{proof} $\square$

{}From the example we also see that $\alpha \neq 0$ is a necessary condition for the theorem to hold. Were $\alpha =0$ in that special case, there were \emph{no} field equations and the vector bundle morphism $ a \colon T \S \to TM$ were unrestricted and in general not a Lie algebroid morphism. On the other hand, it is not difficult to see (but it certainly also follows from the theorem) that for the (possibly $H$-twisted) PSM the field equations \emph{do not change} when adding the term $S_{geom}$ with some non-vanishing $\a$. We thus propose to consider $S_{DSM}$ for non-vanishing $\a$ in general.\footnote{In \cite{KSS} it is conjectured that the theory with $\a = 0$ is (essentially) equivalent to the theory with $\a \neq 0$ in general. E.g.~for $D=TM$, and $\alpha \neq 0$ the moduli space of solutions to the field equations (maps $X \colon \S \to M$) \emph{up to gauge symmetries} (which at least contain the homotopies of this map $X$) is zero dimensional, like the moduli space for the vanishing action $\a=0$. It is argued that the geometrical part serves as a kind of regulator for the general theory, which also should permit localization techniques on the quantum level.}

There are two more important issues which we need to at least mention in this context. First, if the DSM is to be associated to a Dirac structure on $M$ and all other structures used in the definition of the functional are to be auxiliary, one
needs to show that it is in the end the \emph{cohomology class} of $H$ only that enters the theory effectively. In \cite{KSS} we proved
\begin{proposition}
The DSM action transforms under a change of splitting (\ref{change-of-splitting})
according to \beq S_{DSM}\mto S_{DSM} +\int\limits_{\Sigma}X^* (B)(\md X-V\!\stackrel{\wedge}{,}\!\md X-V)\,.
\eeq
Let $\a\ne 0$ and $B$ be a "sufficiently small" 2-form, then
there exists a change of variables $\bar V=V+\de V$ and $\bar A=A+\de A$, such that
\beqn
S_{DSM}[X, \bar V,\bar A]= S_{DSM}[X,V,A] +\int\limits_{\Sigma}X^* (B)(\md X-V\!\stackrel{\wedge}{,}\!\md X-V)\,.
\eeq
\end{proposition}
So, a change of the splitting can be compensated for by a (local) diffeomorphism on the space of fields.

Secondly, we did not yet touch the issue of the gauge symmetries, neither in the previous section on the AKSZ sigma models nor for the DSMs. Certainly the gauge symmetries are of utmost importance in topological field theories (without them we were never able to arrive at a finite dimensional moduli space of solutions for instance). While the BV-formalism produces them by means of the BV-operator for the AKSZ sigma models (although possibly in a coordinate dependent way, cf., e.g., \cite{BKS} addressing this issue), they are less obvious to find for the DSM. In fact, here also all the \emph{auxiliary structures do enter}, cf.~\cite{KSS} for the corresponding formulas. It is only \emph{onshell}, i.e.~on using the field equations (here only $\md X = V$), that the gauge transformations obtain a nice geometrical interpretation: they turn out to become ``Lie algebroid homotopies'' (cf.~also \cite{BKS}) in this case.  We do, however, not want to go into further details on this here; in the present article we decided to focus more on the field equations, reassuring the reader in words that the more intricate gauge symmetries fit nicely into the picture as well, completing it in an essential way.


\section{Yang-Mills type sigma models}
In the Introduction we recapitulated the idea of sigma models: one wants to replace the flat target space $\R^n$ of usually a collection of $n$ scalar fields (functions on spacetime $\Sigma$) by some geometrical object, like a Riemannian manifold. We posed the question, if, in the context of gauge fields (1-forms on space time), we can replace in a likewise fashion the ``flat'' Lie algebra $\R^n$ (or, more generally, $\g$) by some nontrivially curved geometrical object. In fact, the Poisson sigma model (\ref{PSM}), or more generally, the AKSZ sigma model (\ref{AKSZSM}), provides, in some sense, half a step into this direction: Let us consider the three dimensional case, where this is most evident. The quadratic Lie algebra $\g$ needed for the definition of the Cherns Simons gauge theory can be generalized to a Courant algebroid, associated to which is the Courant sigma model \cite{CS-Ikeda, RoytenbergAKSZ}, which when specializing to the ``flat case''  $\g$ reproduces the Chern Simons theory. While these models realize in a geometrically nice way the right (classical) field content, the target algebroid being represented by a PQ-manifold and the Lie algebra valued 1-forms of a Yang-Mills theory being interpreted as and generalized to degree preserving maps from $T[1]\Sigma$ to the respective target, and also the gauge transformations are generalized in a reasonably looking way, there is, from the \emph{physical} point of view, a major drawback or ``flaw'' of these theories: What mathematically is usually considered an advantage of a field theory, namely to be topological, in the context of physics rules out a theory for being feasable to describe the degrees of freedom we see realized in the interacting world around us.

Indeed, the space of flat connections, which are the field equations of the Chern Simons theory, modulo gauge transformations is (for, say, $\Sigma$ without boundary and of finite genus) a finite dimensional space and this generalizes in a likewise fashion to the moduli space of solutions modulo gauge transformations for all the AKSZ models, where one considers the space of Q-morphisms modulo Q-homotopy. The moduli space needed to host a physical particle (like a photon, electron, Higgs, etc), on the other hand, is always infinite dimensional  (like the space of harmonic functions on $\Sigma$ for a Laplacian corresponding to a  (d--1,1)--signature metric). There is also another way of seeing that one has gone half way only: Such as we want that when in a sigma model for scalar fields  the choice of a ``flat background'' (i.e., in that case, that the target is a flat Riemannian manifold $\R^n$) the action reduces to (\ref{eq:freescalar}), we want that when in a Yang-Mills type sigma model the target is chosen to be a Lie algebra $\R^n$ or, more generally, $\g$,  the gauge theory reduces to  (\ref{eq:free1forms}) and (\ref{eq:YM}), respectively. So we will pose this condition, maybe adding that there should be a ``comparable number'' of gauge symmetries in the general model as in the special case.

In a similar way we may generalize this condition by extending it to also higher form gauge fields $B^b$, namely that for an appropriate flat choice of the target geometry one obtains from the (higher) Yang-Mills type sigma model
\beq
S[B^b] =  \frac{1}{2} \int_\S \rd B^b \wedge *   \rd B^b
\, . \label{eq:free2forms}
\eeq
For the case that $B$s are 2-forms, this will yield an action functional  for nonabelian gerbes.

Before continuing we bring the standard Yang-Mills action (\ref{eq:YM}) into a form closer to a topological model first. We consider a trivial bundle in what follows, in which case the curvature or field strength is a $\g$-valued 2-form. If $d$ denotes the dimension of spacetime $\Sigma$, we introduce a $\g^*$-valued $d\!-\!2$ form $\Lambda$ in addition to the connection 1-forms $A$. With $\kappa^{-1}$ denoting the scalar product on $\g^*$ induced by $\kappa$ on $\g$, then
\beq  S'_{YM}[A,\Lambda] := \int_\S \langle \Lambda , F \rangle + \kappa^{-1}(\Lambda \! \stackrel{\wedge}{,} \! *\Lambda) \label{eq:YMprime} \eeq
is easy to be seen as equivalent to  (\ref{eq:YM}) on the classical level (Euler Lagrange equations) after elimination of the auxiliary field $\Lambda$. (Equivalence on the quantum level follows from a Gaussian integration over the field $\Lambda$, as it enters the action quadratically only---up to an overall factor, which is irrelevant for the present considerations).

The first part of this action is topological, it is only the second term, breaking some of the symmetries of the topological one, that renders the theory physical.
Let us first generalize the topological part, a so-called BF-theory, by an appropriate reformulation. In fact, we can obtain the BF-theory from the AKSZ-method by
considering $\CM_2 = T^*[d\!-\!1] \g[1]$ as a target QP-manifold (as before d is the dimension of spacetime $\Sigma$). Indeed, it is canonically a symplectic manifold with the symplectic form being of degree $d\!-\!1$, as it was required to be in section \ref{AKSZ}, and the Q-vector field on $\g[1]$, eq.~(\ref{CEQ}), can always be lifted canonically to the cotangent bundle with the folllowing Hamiltonian \footnote{We use the following sign convention: $\{p_a, q^b\}=(-1)^{m\mid q^b\mid}\de_a^b$, where $m$ is the degree
of a symplectic form.} 
\beq \CQ = (-1)^{d} \frac{1}{2}C_{bc}^a \theta^b\theta^c p_a \, .
\eeq
The AKSZ action (\ref{AKSZclass}) now just produces
\beq S_{BF} [\Lambda, A] = \int_\Sigma \Lambda_a (\md A^a + \frac{1}{2} C^a_{bc} A^b \wedge A^c) \,  \eeq
where we used $A^a = a^* (\theta^a)$ and $\Lambda_a = a^*(p_a)$; and this is just equal to the first term in (\ref{eq:YMprime}). It is now easy to find a generalization of this part of the action for a general ($\N_0$-graded) Q-manifold $(\widetilde \CM_2,\widetilde Q_2)$: Just consider $\CM_2 = T^*[d\!-\!1] \widetilde \CM_2$ as a target QP-manifold with the canonical lift of $\widetilde Q_2$ to the Hamiltonian vector field $Q_2$ (with Hamiltonian $\CQ_2=\CQ$).

Let us exemplify this at the example of a Lie algebroid, where $\widetilde \CM_2=E[1]$ and $\widetilde Q_2$ is given by means of formula (\ref{actionQ}) (with $\rho_a=\rho_a^i \partial/\partial x^i)$ and $\rho^i_a(x)$, $C^c_{ab}(x)$ being the Lie algebroid structure functions in the local coordinates $x^i, \theta^a$ on $E[1]$). Now the same procedure yields
the action functional for a ``Lie algebroid BF-theory'' \cite{gravityfromLiealgebroidmorphisms,AYM,BonZab}
\beq S_{LABF} [\Lambda_i, \Lambda_a,X^i, A^a] = \int_\S (\Lambda_i \wedge F^i + \Lambda_a \wedge F^a) \, \label{LABF}
\eeq
where $\Lambda_i$ and $\Lambda_a$ are $d\!-\!1$ forms and $d\!-\!2$ forms, respectively, and
\begin{eqnarray} F^i &= &\md X^i - \rho^i_a(X) A^a \label{Fi}\\
F^a &=& \md A^a + \frac{1}{2} C^a_{bc}(X) A^b \wedge A^c \label{Fa}
\end{eqnarray}
are the generalizations of the YM-curvatures, which we prefer to call \emph{field strengths} in the general case. We use the notation $A^\alpha= a^*(q^\alpha)$, where $q^\a$ are graded coordinates on $\widetilde \CM_2$ and $\Lambda_\alpha = a^*(p_\alpha)$, the corresponding momenta (or anti-fields) on $\CM_2=T[1]\widetilde \CM_2$. Note that the map $a \colon T[1] \S \to \CM_2$ is degree preserving, so that $X^i : \equiv A^i = a^*(x^i)$ are 0-forms on $\Sigma$ or scalar fields, $A^a = a^*(\theta^a)$ are 1-forms, and $\Lambda_i$ and $\Lambda_a$ are $d\!-\!1$ and $d\!-\!2$ forms, respectively.

We remark in parenthesis that although (\ref{LABF}) is inherently coordinate and frame independent, the expression (\ref{Fa}) is not. It is the splitting of the field strength(s) into two independent parts (\ref{Fi}) and (\ref{Fa}), which cannot be performed canonically. One way of curing this is by introducing an additional connection on the (target) Lie algebroid $E$, cf., e.g., \cite{BKS,AYM}. Better is to realize that the whole information is captured by the map (\ref{f}), taking values in the \emph{tangent} bundle over the target $\CM_2=T^*\widetilde \CM_2$. We now turn to this perspective.

The map $f$ defined in (\ref{f}) covers the map $a \colon \CM_1 \to \CM_2$ (it is also this fact that makes the difference in (\ref{f}) well-defined, being the difference between two elements in the fiber, which is a vector space, over the same point)

\begin{displaymath}
    \xymatrix{ & T[1]\cM_2 \ar[d]\\
               \cM_1\ar[ur]^f\ar[r]_a & \cM_2  }
\end{displaymath}



\noindent and measures the deviation of $a$ to be a Q-morphism, in which case the following diagram would commute

\begin{displaymath}
    \xymatrix{ T[1]\cM_1\ar[r]^{a_*} & T[1]\cM_2 \\
               \cM_1\ar[u]^{Q_1}\ar[r]_a & \cM_2 \ar[u]^{Q_2} }
\end{displaymath}

Denote by $x^A = (q^\alpha,p_\alpha)$ the graded coordinates on $\CM_2=T^*[d\!-\!1]\widetilde \CM_2$ and by $\md x^A$ the induced fiber-linear coordinates on $T [1]\CM_2$. To use the map $f$ in practise, the following lemma is helpful:
\begin{lemma} \label{fstar}
Let $(x^A, \md x^A)$ be local coordinates on $T [1]\CM_2$, $a \in \Mor(T[1]\S,\CM_2)$, and $f$ the map defined in (\ref{f}). Then
\beq f^* (x^A) = a^*(x^A) \qquad f^*(\md x^A) = (\md a^* - a^* Q_2 )  x^A \, . \eeq
\end{lemma}
\begin{proof} The first part is evident from $f$ covering $a$ as remarked above already. The second part is rather straightforward and we show only one first step in the calculation. The vector field $Q_2$ of degree 1 is a \emph{degree preserving} section of the tangent bundle over $\CM_2$, if the latter is shifted by 1 in degree, i.e., as here, one considers $T[1]\CM_2$. The section maps a point with coordinates $x^A$ to the point $(x^A,Q_2^A(x))$. Thus we find for the pull-back of $\md x^A$ by this section: $(Q_2)^* \md x^A=Q_2^A(x)$, which, however, can also be rewritten as $Q_2$ applied to $x^A$. The remaining similar steps we leave as an exercise to the reader. We just remark that we had used $Q_1 = \md d$ and that certainly pull-backs go the reverse order as the respective map, cf the second diagram above.
\end{proof} $\square$

We remark also that $\md x^A$ is considered as a \emph{function} on $T [1]\CM_2$ and is pulled back as such, \emph{not} as a differential form on $T[1]\CM_2$, while a function on $T[1]\CM_2$ is, on the other hand, equally well a differential form on $\CM_2$. In fact, the so-understood pull back will \emph{not} be a chain map w.r.t.~the  de Rahm differential (a vector field on $T[1]\CM_2$), but only w.r.t.~a modified vector field on $T[1]\CM_2$, cf.~Prop.~\ref{chain} below.

We see that the field strengths $F^i$ and $F^a$ are nothing but the pull back of $\md x^i$ and $\md \theta^a$ by $f$ and thus find
\begin{cor} \label{pdq} The action functional (\ref{LABF}) can be written more compactly as
\beq S_{LABF}[a] = \int_\S f^* (p_\a \md q^\a) \eeq
where $a \colon T[1]\Sigma \to \CM_2=T^*[d\!-\!1]E[1]$ and $p_\a \md q^\a$ is the canonical 1-form on $\CM_2$.
\end{cor}

Now we are in the position to generalize the Yang-Mills action functional (\ref{eq:YM}) or better (\ref{eq:YMprime}) to the Lie algebroid setting, i.e.~to a situation where the structural Lie algebra $\g$ is replaced by a Lie algebroid $E$ (cf.~definition \ref{Lodayalgebroid}). A Lie algebroid Yang-Mills theory is defined as a functional on $a \in \Mor(T[1]\Sigma,\CM_2=T^*[d\!-\!1]E[1])$ and is of the form \cite{AYM}
\beq S'_{LAYM}[a] = S_{LABF}[a] + \int_\S \Eg^{-1}(\Lambda_{(d-2)}  \! \stackrel{\wedge}{,} \! *\Lambda_{(d-2)}) \label{eq:LAYMprime}\eeq
where $\Lambda_{(d-2)} \in \Omega^{d-2}(\Sigma,X^*E)$ are the Lagrange multiplier fields of the 2 form field strengths only. Equivalently, eliminating precisely those Lagrange multipliers like in the transition from (\ref{eq:YMprime}) to (\ref{eq:YM}), we can regard the somewhat more explicit action functional
\beq
S_{LAYM}[A^a,X^i,\Lambda_i] = \int_\S \Lambda_i \wedge F^i + \frac{1}{2} \Eg(X)_{ab} F^a \wedge * F^b  \, ,\eeq
where $F^i$ and $F^a$ are the field strengths (\ref{Fi}) and (\ref{Fa}), respectively, detailed above. In a standard YM theory, the metric on the Lie algebra needs to be ad-invariant. Here we can ask that $(E,\Eg)$ define a maximally symmetric $E$-Riemannian space (cf.~Def.~\ref{homo}) or that the fiber metric is invariant w.r.t.~a representation (\ref{Enabla}) induced by any auxiliary connection on $E$. (For still less restrictive conditions cf.~\cite{Mayer}).

This is now easy to generalize to higher form degrees of the gauge field. Let us consider a tower of gauge fields $X^i, A^a, \ldots , B^B$, which are 0-forms, 1-forms etc, respectively, up to a highest degree $p$, so $B^B$ being p-forms. The role of the structural Lie algebra will now be played by a degree p Q-manifold $(\widetilde \CM_2,Q_2)$. The gauge fields are collected into a (degree preserving) map $a \colon T[1]\S \to \widetilde \CM_2$. This is extended to $a \in \Mor(T[1]\S,\CM_2)$ with $\CM_2=T^*[d\!-\!1] \widetilde \CM_2$ the corresponding PQ-manifold with canonical 1-form $p_\a \md q^\a$. We then consider \cite{talk}
\beq S'_{higherYM}[a] = \int_\S f^* (p_\a \md q^\a)  + \int_\S \Vg^{-1}(\Lambda_{(d-p-1)}  \! \stackrel{\wedge}{,} \! *\Lambda_{(d-p-1)}) \,  \label{eq:higherLAYMprime}\eeq
where $\Vg$ is a fiber metric on the vector bundle $V \to M$ corresponding to
the degree $p$ variables on $\CM_2$ (in fact, there is a canonical quotient of $\CM_2$ yielding $V[p]$). Similarly to before we can also eliminate the lowest form degree Lagrange multipliers $\Lambda_{(d-p-1)})$ in $S'_{higherYM}$ to obtain $S_{higherYM}[a]$ in which the highest form degree field strength is squared by means of the fiber metric $\Vg$.

There is always a condition on this fiber metric generalizing the ad-invariance in the standard YM case. Let us specify this condition in the case $p=2$. We were mentioning in section \ref{Qman} that a degree 2 Q-manifold corresponds to a Lie 2-algebroid, cf.~Def.~\ref{Lie2} and the ensuing text. In this case the condition to be placed on $\Vg$ is its invariance w.r.t.~the $E$-connection $\Ena$ on $V$ \cite{talk}:
\beq \Ena \, \Vg = 0 \, . \eeq
For $p=2$ this is a possible definition of an action functional for nonabelian gerbes.

In this contribution we never discussed the gauge transformations in any detail. In a gauge theory they are certainly of utmost importance and much further motivation for the  theories discussed here come from looking at their gauge invariance also.

We make here only two general remarks on the gauge transformations:
First, for a general algebroid type gauge theory the gauge symmetries have a ``generic part'' that comes from the structural algebroid (possibly including some particularly important geometrical ingredients from the target like a compatible symplectic form), but there can be also contributions to them that show all possible structural ingredients used in the definition of the action functional. An (impressive) example for the latter scenario is provided by the Dirac sigma model, cf.~section \ref{DSM} and \cite{KSS}. Second,
for the higher Yang-Mills action functionals proposed above, or also any other \emph{physical} gauge theory proposed in this setting, the BV-formulation, which also captures elegantly the gauge symmetries of a theory, will \emph{not} just consist in replacing the morphisms $a$ by supermaps $\widetilde a$ as it was the case with the AKSZ type sigma models. In fact, the space $\Maps(T[1]\Sigma, \CM_2)$ is too big: It contains ghosts and ghosts for ghosts etc for gauge symmetries of the BF-part of the action that are broken by the fiber metric part. It is tempting, however, to try to obtain the BV-formulation of the full physical model by a kind of supersymplectic reduction induced by the symmetry breaking terms.

In the present article we used the framework of the AKSZ sigma models to develop the field content of the higher gauge theories and, by the material at hand from the previous sections, it was natural to require the structural Lie algebra to be replaced by a general Q-manifold. This is in so far a physical approach as that the notion of Q-manifolds arose in a physical context, namely to describe the supergeometric framework underlying the BV-formulation \cite{Schwarz_semiclassical}. In this way,  the generalization of the connection
1-forms of a Yang Mills theory to the field content considered
sounds much less compulsory than it really is. As mentioned in the introduction, physical considerations often suggest geometrical notions and generalizations into very particular directions. This is also the case in the context of ``higher gauge theories''---that is gauge theories where the field content is required to consist locally of a tower of differential forms up to some highest degree $p$. It is evident that then locally we can introduce (in a canonical way) a graded manifold $\CM_2$ of degree $p$ and interpret the field content as a morphism from $T[1]\S$ to $\CM_2$. Even without yet considering gauge transformations acting on these fields, there are sensible requirements on the properties of the theory related to a generalization of the Bianchi identities in ordinary gauge theories that \emph{require} that the target $\CM_2$ has to be a Q-manifold. This approach would have fitted the logic of the present article very well, too, but it would have extended its length further as well so that we decided to not reproduce these considerations, which were presented in \cite{talk} and can be found in \cite{GruetzmannStrobl}.

Up to now we only considered the generalization of gauge theories in the context of trivial bundles. In principle, it is not at all clear that the local construction generalizes in a straightforward way to a global one. But fortunately here this is indeed the case. We briefly sketch how this works, since, from the \emph{mathematical} point of view, much of the interest in such gauge theories would precisely stem from the global picture. For further details we have to refer the reader to \cite{Q-bundles}.

We again start from ordinary Yang Mills theories. There the space of fields are not Lie algebra valued 1-forms on our spacetime $\S$, but rather connections in a principal $G$-bunlde $P$ over $\S$. Associated to each principal bundle, however, there is the so-called Atiyah algebroid, which is a very particular Lie algebroid over $\Sigma$. As a total space it is equal to $TP/G$, which is canonically a Lie algebroid by use of the de Rahm differential on $TP$. The image of its anchor map is $T \Sigma$. Shifting the fiber degrees by 1, this becomes a bundle in the category of Q-manifolds $T[1]P/G \to T[1]\Sigma$, the typical fiber of which is the Q-manifold $\g[1]$ (as before, $\g$ is the Lie algebra of the structure group $G$). We call this a Q-bundle. Moreover, a connection on $P$ is in bijection to a section of this bundle \emph{viewed} as a bundle of graded manifolds (it is the \emph{flat} connections that are the sections of the Q-bundle).

The generalization is obvious now: As total spaces we consider Q-bundles $\pi \colon (\CM,Q) \to (\CM_1,Q_1)$, where as the base $(\CM_1,Q_1)$ we keep $(T[1]\S, \md)$ for simplicity, but as typical fibers we permit any (non-negatively graded) Q-manifold $(\CM_2,Q_2)$.
Gauge fields are any degree preserving maps $a \colon \CM_1 \to \CM$ such that
$\pi \circ a = \mathrm{id}_{\CM_1}$. Many constructions that were done for
gauge fields $a \colon \CM_1 \to \CM_2$, like the definition of the field strength $f$, can now be repeated without modification by replacing merely $\CM_2$ by the total space $\CM$. The (generic part of the) gauge symmetries receive the nice interpretation of (possibly a graded subgroup of) vertical inner automorphisms of this bundle. And the gauge invariance of action functionals like
(\ref{eq:higherLAYMprime}) permits to glue local expressions for functionals together to global ones.

The global setting certainly also opens gates for studying characteristic classes for the bundles. In particular, there exists a generalization of the Chern-Weil formalism. In the classical setting, a connection on a principal bundle together with an ad-invariant polynomial on the Lie algebra gives a cohomology class on the base manifold, which does not depend on the connection and which vanishes when the bundle is trivial. In the context of the classical construction one considers the so-called Weil algebra $W(\g)=S^\cdot (\g^*)\otimes\Lambda (\g^*)$ with a certain differential $Q_W$ on it---the connection on $P$ then induces a map from a subcomplex of so-called $G$-basic elements of $W$ to the de Rahm complex on $\S$.

To generalize this construction, one may first observe that $W$ can be identified with differential forms on $\Omega(\g[1])$, which we can identify with functions
on $T[1]\g[1]$. This reminds one of the target of the map $f$, namely when $\CM_2 = \g[1]$, cf. the first diagram above. Given a general Q-manifold $\CM_2$ and a degree preserving map $a \colon \CM_1 \to \CM_2$, can we equip $T[1]\CM_2$ with a Q-structure $Q_{T\CM_2}$, such that $f \colon (\CM_1,Q_1) \to  (T[1]\CM_2,Q_{T\CM_2})$ is \emph{always} a Q-morphism?
The answer is affirmative \cite{Q-bundles}:
\begin{prop} \label{chain} The map (\ref{f}) is a Q-morphism, if we equip $T[1]\CM_2$ with the Q-structure $Q_{T\CM_2}= \md + \CL_{Q_2}$, where $\md$ is the de Rahm differential  on $\CM_2$ and $\CL_{v}$ the Lie derivative w.r.t.~a vector field $v$  on $\CM_2$, both viewed as degree 1 vector fields on
$T[1]\CM_2$: \beq Q_1 f^* = f^* Q_{T\CM_2} \, . \label{fchain} \eeq
\end{prop}
Indeed, the canonical lift $Q_{T\CM_2}$ of the differential $Q_2$ to its tangent reproduces the Weil differential $Q_W$ mentioned above for $\CM_2 = \g[1]$.
One now generalizes also the notion of being basic, where in the general setting a graded gauge or holonomy group $\CG$ plays an important role. In the end one obtains \cite{Q-bundles} a generalization of the Chern-Weil map in the form
\begin{theorem} \label{ChernWeil} Let $\pi
\colon \cM\to T[1]\S$ be a Q-bundle with a typical fiber $\cM_2$, a  holonomy
group $\CG$, and $a$ a ``gauge field'', i.e.~a section of $\pi$ (in the graded
sense). Then there is a well-defined
map in cohomology
\beq\label{characteristic_map} H^p (\Omega (\cM_2)_\CG,
Q_{\scriptscriptstyle T\cM_2})\to H_{deRahm}^p(\S)\;,\eeq
 which does not depend on homotopies of $a$.
\end{theorem}

We want to end this section and thus this article by commenting on a rather intriguing link of this construction with the topological sigma models we were considering in a previous section, namely section \ref{AKSZ}. There is a famous relation of a characteristic class on a principal bundle with a quadratic structure group, the second Chern class or the first Pontryagin class, with the integrand of the Chern-Simons gauge theory: locally $``\mathrm{Second\;
Chern\; form}=\md \left(\mathrm{Chern\!-\!Simons\;
form}\right)''$. This class results from the Chern-Weil formalism by applying  the map to the ad-invariant metric $\kappa$ on $\g$, viewed as quadratic polynomial. In sec.~\ref{Qman} we learnt that $\kappa$ can be viewed as a symplectic form $\omega$ of degree 2 on $\g[1]$. Locally, the above generalized Chern-Weil map is nothing but $f^*$. On the other hand, the AKSZ sigma model resulted from choosing a QP-manifold $(\CM_2,Q_2,\omega)$ as a target---and for the Chern-Simons theory it is this $\omega$ that corresponds to $\kappa$.

We are thus lead to ask for a relation of the AKSZ sigma models with $f^* \omega$ for a general QP-manifold (that would serve as a typical fiber in a Q-bundle). In Cor.~\ref{pdq} we found that the integrand of the AKSZ sigma model can be written as the pullback of the canonical 1-form by $f$ canonical 1-form, if $\CM_2$ is the cotangent bundle $\CM_2 = T^*[d\!-\!1]\widetilde \CM_2$ with $\widetilde \CM_2$ a degree 1 Q-manifold.\footnote{In fact, this is true for a general Q-manifold, as one can prove directly in generalization of Cor.~\ref{pdq} or deduce from Theorem \ref{AKSZchar} below.}  We thus compute $\md f^*(p_\alpha \md q^\alpha) = f^* (\md + \CL_{Q_2}) p_\alpha \md q^\alpha$, where we used the chain property (\ref{fchain}). But the canonical 1-form is invariant w.r.t.~the Hamiltonian lift $Q_2$ of a vector field $\widetilde Q_2$ on its base $\widetilde \CM_2$; so only the first term remains on the r.h.s., which indeed has the form $f^* \omega$. In fact, this statement that we proved for $\CM_2 = T^*[d\!-\!1]E[1]$, is true for any QP-manifold; one has \cite{Q-bundles}
\begin{theorem} \label{AKSZchar}
Consider a PQ-manifold $\CM_2$ with symplectic form $\omega$ of positive degree $d-1$ and a $(d\!+\!1)$-dimensional
manifold $N$ with boundary $\partial N=\Sigma$, and $a \in \Mor(T[1]N,\CM_2)$. Then
\beq 
\int_{N} f^* \omega =  S_{AKSZ}[a]\, ,\eeq
the AKSZ sigma model (\ref{AKSZclass}) on $\Sigma$.
\end{theorem}
As mentioned, these observations generalize the relation of the Pontryagin class to the topological Chern-Simons theory to a much wider range of bundles and topological models.

Certainly, in some cases Q-manifolds and Q-bundles may be ``integrable'' (already not any Lie algebroid is integrable to a Lie
groupoid---cf.~\cite{CrainicFernandes} for the necessary and sufficient conditions for the integrability of Lie algebroids).
Similar to the relation of the Atiyah algebroids to principal bundels,
integrable Q-bundles will will correspond to rather intricate other type of bundles---in the case of the Lie algebroid Yang-Mills theories, e.g., these are bundles with typical fiber being a groupoid, in the case of the fibers being degree 2 Q-manifolds over a point these will be the nonabelian bundle gerbes of \cite{BreenMessing}. What we called gauge fields  and field strengths and others might call ``connections'' and ``curvatures'', respectively, are often rather intricate to describe in the integrated pictures---in contrast to what happens in the Q-bundle setting.

It is our impression that the Q-bundle picture has several advantages in the description of higher gauge theories. Also, they exist and give rise to potentially interesting action functionals even if the underlying Q-manifolds are not integrable; a topological example of this kind is the Poisson sigma model, which exists and has its merits even if the target Poisson manifold $M$ (the Lie algebroid $T^*M$) is not integrable.

In any case, we hope to have convinced the reader who was not yet familiar with all the concepts used in this article, that algebroids of various kinds are interesting as mathematical objects on the one hand, interpolating in a novel way between various established notions like tensor fields on manifolds and Lie algebras or
symplectic and complex geometry, and, on the other hand, that they appear in
many different facets in the context of sigma models and higher gauge theories,
in fact in such a way that they have to be considered indispensable for an elegant and technically and conceptually convincing discussion of such theories.

\section*{Acknowledgement}
We are deeply indebted to V.~Cortes for his incredible patience with us in finishing this contribution.


\end{document}